\documentclass[final,oneside,notitlepage,10pt,letterpaper]{article}
\usepackage{iftex}
\ifPDFTeX\pdfoutput=1\fi
\usepackage[T1]{fontenc}
\usepackage{lmodern}
\usepackage[english]{babel}
\usepackage[left=3cm,right=3cm,top=3cm,bottom=3cm]{geometry}
\usepackage{amsmath,amssymb,amsthm,mathtools}
\usepackage{microtype}
\usepackage{enumitem}
\usepackage{etoolbox}
\usepackage{fancyhdr}
\usepackage{xcolor}
\usepackage[hypertexnames=false]{hyperref}
\definecolor{linkblue}{RGB}{0,0,120}
\hypersetup{colorlinks=true,linkcolor=linkblue,citecolor=linkblue,urlcolor=linkblue,linktocpage=true}
\renewcommand{\baselinestretch}{1.05}
\setlist{itemsep=0.3ex,topsep=0.5ex}
\newtheoremstyle{paperthm}{1em}{1em}{\itshape}{0pt}{\bfseries}{.}{.5em}{}
\newtheoremstyle{paperdef}{1em}{1em}{\normalfont}{0pt}{\bfseries}{.}{.5em}{}
\theoremstyle{paperthm}
\newtheorem{theorem}{Theorem}[section]
\newtheorem{proposition}[theorem]{Proposition}

\newtheorem{corollary}[theorem]{Corollary}
\theoremstyle{paperdef}

\newtheorem{remark}[theorem]{Remark}
\numberwithin{equation}{section}
\DeclareMathOperator{\Ad}{Ad}
\DeclareMathOperator{\im}{im}
\DeclareMathOperator{\Pair}{Pair}
\DeclareMathOperator{\tr}{tr}
\newcommand{\R}{\mathbb R}
\newcommand{\U}{\mathrm U}
\newcommand{\SU}{\mathrm{SU}}
\newcommand{\fg}{\mathfrak g}
\newcommand{\fh}{\mathfrak h}
\newcommand{\ff}{\mathfrak f}
\newcommand{\fa}{\mathfrak a}
\newcommand{\F}{\mathcal F}
\newcommand{\Hk}{\mathcal H_\kappa}
\newcommand{\Gadj}{\mathscr G_{\mathrm{adj}}}
\newcommand{\Gapt}{\mathscr G_{\mathrm{ad}}}
\newcommand{\disc}{\mathrm{disc}}
\newcommand{\dH}{d_{\mathrm H}}
\usepackage{mathrsfs}

\title{Simplifying 2-Bundle Descriptions of Invariants and Local Dynamics}
\author{Ferreira Bastien}
\date{}

\makeatletter
\patchcmd{\l@section}{1.0em}{0.35em}{}{\PackageWarning{manuscript}{Contents spacing patch failed}}
\makeatother

\begin{document}
\maketitle

\vspace{-5mm}
\begin{abstract}
A string couples to a 2-form gauge potential whose local descriptions need not fit together as a single global form. Principal 2-bundles provide a geometric description of this gluing. However, we recover surface amplitudes and quantum partition functions, as well as local classical field dynamics, without the 2-bundle formalism. {For surface amplitudes, differential characters preserve curvature and flat phases. Ordinary fake-flat crossed-module connections with a fixed flat quotient and a reference lift lead to abelian differential cocycles with local coefficients. In three-dimensional finite gauge theories with cohomological actions, Fourier transformation evaluates the 2-form field sum and gives a sum of ordinary finite-group gauge partition functions. A four-dimensional quadratic example gives a signature response that admits direct cochain evaluation although replacement by an ordinary finite-group gauge theory fails. Established bundle, operator, and module descriptions provide further comparisons for topological and index invariants and twisted-sheaf moduli. For finite-dimensional smoothly separable central crossed modules, an invertible local substitution removes a pure gauge 1-form from invariant classical actions for any adjustment, leaving an ordinary quotient connection and two tensor fields. One of these tensor fields can be separated globally from the remaining lift data after fixing the quotient connection and choosing a smooth bundle splitting. For compact or semisimple quotients, an equivariant splitting makes this separation compatible with the full gauge symmetry without changing the adjustment.}
\end{abstract}

\vspace{-5mm}
{\setcounter{tocdepth}{1}\setlength{\parskip}{0pt}\tableofcontents}

\section{Introduction}
\label{sec:introduction}

A string in a background 2-form field acquires a phase that depends on its worldsheet. If the potential \(B\) is globally defined and \(f:\Sigma\to X\) is the worldsheet map, the phase is \(\exp(2\pi i\int_\Sigma f^*B)\). On a topologically nontrivial background, local potentials must instead be joined by 1-forms on double overlaps and functions on triple overlaps. An abelian gerbe with connection organizes precisely these data \cite{MackaayPicken,AschieriJurco2004}. It is a basic example of a principal 2-bundle, the broader geometric object studied here \cite{RistSaemannWolf2026,Waldorf2026}. The resulting phase can change a spectrum or cancel an anomaly, so the global data are physically consequential \cite{VafaWitten1995,FreedWitten1999}. {The question is whether obtaining these effects requires the 2-bundle presentation, and the same question arises across more general physical and mathematical topics.} In each setting, we ask whether the higher presentation supplies information or an operation absent from the original description, or whether the stated result follows directly from cocycles, ordinary bundles, operators, modules, or a change of fields. Across the applications studied here, these descriptions suffice. In several cases, their equivalence with the higher geometric description is itself a theorem \cite{MackaayPicken,Waldorf2016,KapustinThorngren2017}. {An equivalent presentation preserves the required data and operations. To decide whether it also simplifies their use, we identify what the alternative actually does, from evaluating a phase to eliminating a field.}

The example above shows how this comparison works. We can calculate the surface phase by integrating the local forms and multiplying their transition factors. Equivalently, a differential character is a multiplicative function on 2-cycles satisfying the appropriate Stokes relation. It includes flat phases that the curvature alone would miss and determines the connection class of the abelian 2-bundle \cite{AschieriJurco2004}. The global information is therefore complete before the geometric object is reconstructed from it. The later examples apply the same standard while preserving the compatibility conditions and operations on which each result depends.

{For a general crossed module, the quotient connection of an ordinary fake-flat connection is flat, but its parallel transport can act nontrivially on the abelian kernel. We fix that quotient and subtract a reference lift. The remaining data are differential cocycles with local coefficients, and their closed-surface evaluation lies in the kernel's coinvariant group. This gives a concrete specialization of existing nonabelian holonomy and differential cohomology \cite{SchreiberWaldorfHolonomy,Parzygnat2015,SatiSchreiber2026}. The transported coefficients explain why an ordinary scalar differential character alone need not preserve these data. Section~\ref{sec:twisted-surface} derives the reduction directly from the gluing equations and explains the choices needed to reconstruct marked transport.}

For fixed backgrounds, we ask whether the proposed higher invariant contains information absent from the original description. In the Wess--Zumino--Witten model, the phase in the multiplication law detects the obstruction to a compatible Chern--Simons construction. Winding numbers already compute the same obstruction \cite{GawedzkiWaldorf2009}. In a covering-group construction of magnetic flux, the surface charge is fixed by the original bundle's equatorial transition function \cite{Parzygnat2015}. These comparisons depend on global compatibility and gluing data. Curvature alone need not preserve the information being compared. Band topology and duality admit similar comparisons. {The Dixmier--Douady number in a specified band model agrees with its winding number up to orientation,} and an experiment obtains the corresponding charge from ordinary Berry curvature or the quantum metric \cite{PalumboGoldman2019,ChenEtAl2022}. For mapping-space T-duality, localization determines the proposed maps and characteristic classes from data on the original target \cite{BouwknegtEvslinMathai2004,HanMathai2015,HanMathai2026}. Two proposed band-theory extensions fail in their defining formulas. The proposed wedge powers of a scalar 3-form vanish from the second power onward \cite{PalumboGoldman2019}. Differentiating the line integral in Ref.~\cite{JankowskiSlagerPalumbo2025} produces a curvature term omitted from its reconstruction formula.

The analytic index gerbe records a family of fermionic operators through determinant lines of finite spectral subspaces and eta forms \cite{Lott2002}. These data give the complete differential cocycle, and its transgression is an ordinary determinant line with an adiabatic connection \cite{Bunke2002}. Algebraic geometry provides a parallel comparison: in the arithmetic period--index problem, twisted sheaves can be expressed as modules over a chosen Azumaya algebra, preserving families and deformation groups \cite{Lieblich2008}. The difficult existence theorem for modules of sufficiently small size remains necessary \cite{Lieblich2015}. Replacing the gerbe presentation does not turn either the analytic or the arithmetic theorem into an elementary one.

For a dynamical theory, matching a fixed-background amplitude does not yet compare the field theories. We must follow the action, gauge identifications, and sum over fields. Finite gauge theories make this question physically concrete. They describe long-distance properties of gapped phases, including ground-state degeneracy and the phases acquired when particles and strings are linked \cite{DijkgraafWitten1990,KapustinSeiberg2014,KapustinThorngren2017}. In three dimensions, character orthogonality performs the 2-form field sum even in the presence of a cohomological action. The partition function is a sum of partition functions of ordinary finite-group gauge theories, whose groups and actions we identify explicitly. In four dimensions, the same argument expresses the zero-action partition function as an ordinary gauge partition function multiplied by an Euler factor. {The established duality also allows a linear 2-form action, which can change the dual group from a semidirect product to a nonsplit extension \cite{KapustinThorngren2017}.} A quadratic action can prevent this reduction: we compute an example whose remaining response is a signature phase. These restrictions on the action are part of the dualization in Refs.~\cite{KapustinThorngren2017,KapustinThorngren2013Theta}. The ordinary theory, when it exists, is obtained by evaluating part of the original sum.

Compact abelian BF theory gives the particle--string linking phase by integrating out a field. A higher gauge formulation associated with the Euclidean or Poincar\'e group also has an ordinary BF action differing by a boundary term \cite{GirelliTsimiklis2022}. That term changes which boundary variables are treated as coordinates and momenta, while leaving the bulk equations unchanged. The comparison concerns these explicit actions. The constraints and quantization needed to obtain gravity are further ingredients.

For smooth fields, we study principal 2-bundles through the nonabelian bundle-gerbe model with adjusted connections of Ref.~\cite{Waldorf2026}. We work with the finite-dimensional class defined in Section~\ref{sec:data}. Using its gauge laws, we make an invertible local substitution into an ordinary quotient connection and two tensor fields. One additional 1-form variable is pure gauge. Every local action invariant under its shift descends to the remaining fields, up to a boundary term, with the same local classical equations. {This local reduction allows arbitrary adjustments satisfying the three adjustment identities. Although the local 2-form \(c\) then transforms affinely with the quotient curvature, a smooth bundle splitting gives a global normalization over any fixed quotient connection that admits a lift. For compact or semisimple quotients, an equivariant splitting makes the same adjustment adapted and gives a product compatible with the full gauge symmetry. An explicit shear example admits the fixed-quotient normalization but no equivariant splitting, so these two global conclusions have different hypotheses.} In the adapted variables, a quadratic action describes an ordinary massive 2-form, while more general invariant actions give tensor interactions. The inference from abelian gauge symmetry to a free theory, suggested in the motivating discussions of Refs.~\cite{SaemannSchmidt2020,RistSaemannWolf2026}, fails for these actions: the four-point amplitude is nonzero.

The paper combines established equivalent presentations with explicit reductions and corrections to proposed formulas. Sections~\ref{sec:surface}--\ref{sec:tduality} and Section~\ref{sec:algebra} use transport, localization, index, and module constructions to reproduce the specified results. {Proposition~\ref{prop:crossed-surface} extends the surface comparison to noncentral crossed modules through local coefficients.} The finite Fourier derivations in Section~\ref{sec:finite} implement the dualizations of Ref.~\cite{KapustinThorngren2017}, keeping their action restrictions and normalization. The quadratic example specializes the theory of Ref.~\cite{KapustinThorngren2013Theta}. For smooth fields, the invariance and shift identities of Ref.~\cite{Waldorf2026} give the global decomposition in Theorem~\ref{thm:global}. {Theorem~\ref{thm:fixed-quotient-normalization} removes adaptedness over a fixed quotient connection, while Corollaries~\ref{cor:compact-quotient} and \ref{cor:semisimple-quotient} provide equivariant splittings for the unchanged adjustment.} We make the local substitution and residual transformations explicit, prove the action reduction in Theorem~\ref{thm:actions}{ without assuming adaptedness}, and calculate the resulting dynamics in Section~\ref{sec:physics}. The band-theory counterexamples in Section~\ref{sec:bands} identify failures of particular formulas, separately from the comparison of equivalent presentations. None of the closed-partition-function comparisons is used to infer equivalence of complete quantum theories.

\section{Conventions for connections and descent}
\label{sec:background}

{Nonabelian bundle gerbes provide a geometric realization of principal 2-bundle descent for the Lie crossed modules considered here \cite{RistSaemannWolf2026,Waldorf2026}.} An abelian \(\U(1)\) gerbe is the special case with structure 2-group \(B\U(1)\), represented by the crossed module \(\U(1)\to\{1\}\). {Abelian \emph{bundle 2-gerbes} carry differential data one degree higher than abelian gerbes and occur in Section~\ref{sec:global} as obstructions to lifting ordinary quotient bundles.}

A global 2-group symmetry combines 0-form and 1-form symmetries through their background transformations and belongs to the broader subject of generalized symmetries \cite{BhardwajEtAl2023}. On a right principal \(G\)-bundle \(P\to X\), choose sections with
\begin{align}
&s_j=s_i g_{ij}^{-1},\qquad g_{ik}=g_{jk}g_{ij}
\label{eq:ordinary-transition}
\end{align}
on overlaps. Write \(\fg=\operatorname{Lie}(G)\), let \(\Theta\) be the connection on \(P\), and set \(A_i=s_i^*\Theta\). Then
\begin{align}
&A_j=\Ad_{g_{ij}}A_i-dg_{ij}\,g_{ij}^{-1},\qquad
F_{A_j}=\Ad_{g_{ij}}F_{A_i},\qquad
F_A=dA+\tfrac12[A\wedge A].
\label{eq:ordinary-connection}
\end{align}
Here \([A\wedge A](v,w)=2[A(v),A(w)]\). For a matrix group, \(F_A=dA+A\wedge A\). The Bianchi identity is \(dF_A+[A\wedge F_A]=0\). A horizontal lift \(\widetilde\gamma(t)=s(\gamma(t))U(t)\) of a path \(\gamma:[0,1]\to X\) satisfies
\begin{align}
&\dot U(t)=-A_{\gamma(t)}(\dot\gamma(t))U(t),\qquad U(0)=1.
\label{eq:path-ode}
\end{align}
For a nonmatrix Lie group the right-hand side is interpreted by right translation. A change of section gives \(U(1)\mapsto g(\gamma(1))U(1)g(\gamma(0))^{-1}\). {The path-ordered exponential denotes this solution.} Surface transport below is obtained by successive horizontal lifts and solutions of this ordinary differential equation \cite{LiHolonomy,SchreiberWaldorfHolonomy}.

We write \(\delta\) for the \v{C}ech or cellular coboundary, according to the complex in use. For additive \v{C}ech cochains our sign convention is
\begin{align}
&(\delta c)_{i_0\cdots i_{r+1}}
=\sum_{a=0}^{r+1}(-1)^a c_{i_0\cdots\widehat{i_a}\cdots i_{r+1}}.
\label{eq:cech-differential}
\end{align}
Restriction to the common intersection is implicit. Circle-valued coefficients use products and inverses. Good covers have contractible nonempty finite intersections. The differential-cocycle description combines this coboundary with the exterior derivative. In local-system cochains, coefficient values are transported to a common fibre before summing. These conventions apply respectively to the surface data of Section~\ref{sec:surface} and the finite gauge sums of Section~\ref{sec:finite}.

{
\subsection{Crossed modules and ordinary connections}
\label{sec:ordinary-crossed-module}

A finite-dimensional Lie crossed module consists of Lie groups \(H,G\), a smooth homomorphism \(t:H\to G\), and a smooth action \(\alpha:G\to\operatorname{Aut}(H)\) satisfying
\begin{align}
&t(\alpha_g h)=g\,t(h)\,g^{-1},
\qquad
\alpha_{t(h)}h'=hh'h^{-1}.
\label{eq:peiffer}
\end{align}
We assume smooth separability: \(K=\ker t\) and \(F=G/t(H)\) have their corresponding Lie group structures, and the quotient maps admit local smooth sections. Write \(p:G\to F\) and \(\fa=\operatorname{Lie}(K)=\ker t_*\). The second identity makes \(K\) central in \(H\), hence abelian, and makes the action on \(K\) factor through \(F\). The action of \(F\) on \(K\) may be nontrivial. {We call the crossed module \emph{central} when this action is trivial. We impose centrality only for the adjusted reductions in Section~\ref{sec:data}.}

On a good cover of \(X\), an ordinary connection is described by \(A_i\in\Omega^1(U_i,\fg)\), \(B_i\in\Omega^2(U_i,\fh)\), maps \(g_{ij}:U_{ij}\to G\), 1-forms \(\varphi_{ij}\in\Omega^1(U_{ij},\fh)\), and maps \(h_{ijk}:U_{ijk}\to H\). Here \(\fg,\fh\) are the Lie algebras of \(G,H\). With \(\bar\theta_G\) the right Maurer--Cartan form on \(G\), the double-overlap equations are
\begin{align}
&\begin{aligned}
&A_j+t_*\varphi_{ij}=\Ad_{g_{ij}}A_i-g_{ij}^*\bar\theta_G,\\
&B_j-\alpha_{g_{ij}}B_i
 =-d\varphi_{ij}-\tfrac12[\varphi_{ij}\wedge\varphi_{ij}]
 -\alpha_*(A_j\wedge\varphi_{ij}).
\end{aligned}
\label{eq:ordinary-two-connection}
\end{align}
The action and the brackets include the exterior product. Put \(\widetilde\alpha_h(g)=h^{-1}\alpha_g(h)\), and let \(\theta_H\) be the left Maurer--Cartan form on \(H\). Triple and quadruple overlaps obey
\begin{align}
&\begin{aligned}
&g_{ik}=t(h_{ijk})g_{jk}g_{ij},\\
&\Ad_{h_{ijk}}^{-1}\varphi_{ik}
+(\widetilde\alpha_{h_{ijk}})_*A_k
 =\varphi_{jk}+\alpha_{g_{jk}}\varphi_{ij}
-h_{ijk}^*\theta_H,
\end{aligned}
\label{eq:triple}\\
&h_{ikl}\,\alpha_{g_{kl}}(h_{ijk})=h_{ijl}h_{jkl}.
\label{eq:quadruple}
\end{align}
These are the ordinary connection-preserving descent equations in the right-bundle convention \cite{SchreiberWaldorfHolonomy,Waldorf2026}. Changes of local presentation use the same equations for a local transformation \((g,\varphi)\), together with the transformations \(h\) between gluing choices described in Eq.~\eqref{eq:gauge-two}.

For ordinary crossed-module surface transport we require
\begin{align}
&F_{A_i}=t_*B_i.
\label{eq:ordinary-fake-flat}
\end{align}
This is \emph{fake-flatness}. It ensures that the image under \(t\) of a surface label agrees with its boundary path transport \cite{SchreiberWaldorfHolonomy}. Applying \(p_*\) shows that \(a_i=p_*A_i\) is a flat ordinary \(F\)-connection. The adjusted connections considered later obey different double-overlap laws and need not be fake-flat.
}

\section{Surface holonomy and magnetic charge}
\label{sec:surface}

\subsection{Local potentials and transition functions}

Let \(X\) be a smooth manifold with a good open cover \(\{U_i\}\). Identify \(\R/\mathbb Z\) with the multiplicative circle \(\U(1)\) by \(x\mapsto e^{2\pi i x}\). Abelian gerbe connection data can be written as real 2-forms \(B_i\), real 1-forms \(A_{ij}\), and functions \(g_{ijk}:U_i\cap U_j\cap U_k\to\U(1)\), satisfying
\begin{align}
&\begin{aligned}
&B_j-B_i=dA_{ij},\\
&A_{jk}-A_{ik}+A_{ij}=\frac{1}{2\pi i}g_{ijk}^{-1}dg_{ijk},\\
&g_{jkl}g_{ikl}^{-1}g_{ijl}g_{ijk}^{-1}=1.
\end{aligned}
\label{eq:abelian-cocycle}
\end{align}
The 3-form \(H=dB_i\) is global and has integral periods. {These are the usual differential-cocycle relations \cite{Murray1996,MackaayPicken,AschieriJurco2004}. Changes of presentation specialize those of Section~\ref{sec:ordinary-crossed-module}.} The triple-overlap function supplies the correction needed when the edge connection forms fail to be an ordinary cocycle.

{For \(H=\U(1)\), \(G=\{1\}\), the normalized circle Lie algebra identifies Eq.~\eqref{eq:abelian-cocycle} with Eqs.~\eqref{eq:ordinary-two-connection}--\eqref{eq:quadruple} by \(\varphi_{ij}=-A_{ij}\) and \(h_{ijk}=g_{ijk}^{-1}\). The abelian potentials therefore illustrate the general definition, while fixing the phase convention used below.}

{For a mapped surface subdivided into pieces lying in the \(U_i\), the phase is obtained by integrating \(B_i\) over faces and \(A_{ij}\) over shared edges, and multiplying the transition functions at vertices with their incidence signs.} Changing a face label introduces \(\int dA_{ij}\), canceled by the edge terms. Changing an edge label introduces the logarithmic derivative in Eq.~\eqref{eq:abelian-cocycle}, canceled by vertex factors. The last relation cancels the remaining change of vertex labels. This is the direct holonomy construction in Ref.~\cite{MackaayPicken}.

For a global 2-form \(B\), the choice \(B_i=B\), \(A_{ij}=0\), \(g_{ijk}=1\) gives
\begin{align}
&W_B(f)=\exp\left(2\pi i\int_\Sigma f^*B\right)
\label{eq:global-surface}
\end{align}
for a smooth map \(f:\Sigma\to X\) from a closed oriented surface. On \(\R^3\), \(B=x\,dy\wedge dz\) gives \(W_B(\partial V)=e^{2\pi i\operatorname{Vol}(V)}\). On \(S^2\), a form \(B=\vartheta\omega\) with \(\int_{S^2}\omega=1\) gives \(W_B(S^2)=e^{2\pi i\vartheta}\), although \(H=0\). The second also shows which flat datum a proposed replacement must include.

\subsection{Differential characters and Wess--Zumino amplitudes}

The full set of closed-surface amplitudes has a description without a chosen cover. A {degree-\(3\)} differential character \cite{CheegerSimons1985} is a homomorphism
\begin{align}
&h:Z_2(X)\longrightarrow\U(1),\qquad
h(\partial C)=\exp\left(2\pi i\int_C H\right),
\label{eq:character}
\end{align}
where \(Z_2(X)\) is the group of smooth singular 2-cycles, \(C\) is a smooth singular 3-chain, and \(H\) is a closed 3-form with integral periods. Isomorphism classes of abelian gerbes with connection correspond to these characters \cite{AschieriJurco2004}. Their group is denoted by \(\widehat H^3(X;\mathbb Z)\).

\begin{proposition}
\label{prop:surface-character}
Let \(h\in\widehat H^3(X;\mathbb Z)\) be the differential character associated with an abelian gerbe with connection, and let \(f:\Sigma\to X\) be a smooth map from a closed oriented surface. Its surface holonomy is
\begin{align}
&W_h(f)=h(f_*[\Sigma]).
\label{eq:worldsheet-character}
\end{align}
Here \([\Sigma]\) is any smooth fundamental cycle, and the value is independent of this choice. If \(f_*[\Sigma]=\partial C\), then Eq.~\eqref{eq:worldsheet-character} reduces to the Wess--Zumino filling formula in Eq.~\eqref{eq:character}.
\end{proposition}

\begin{proof}
The character associated with the local data is their surface holonomy. Its boundary value follows from Stokes' theorem and Eq.~\eqref{eq:abelian-cocycle}. The character correspondence identifies background classes with the same function on all 2-cycles. {Two fundamental cycles on \(\Sigma\) differ by a boundary there, and the pulled-back 3-curvature vanishes because \(\Sigma\) is two-dimensional. Equation~\eqref{eq:character} therefore gives the same value for both cycles.} If two fillings are used, their difference is a 3-cycle whose integral of \(H\) is an integer. The exponentiated filling integral is therefore unchanged.
\end{proof}

The exact sequence
\begin{align}
&0\longrightarrow H^2(X;\R/\mathbb Z)
\longrightarrow\widehat H^3(X;\mathbb Z)
\xrightarrow{\mathrm{curv}}\Omega^3_{\mathbb Z}(X)
\longrightarrow0
\label{eq:character-exact}
\end{align}
expresses the difference between specifying the character and specifying only its curvature. Here \(\Omega^3_{\mathbb Z}(X)\) denotes closed 3-forms with integral periods. Flat surface phases, including those associated with torsion characteristic classes, occur in the left-hand term. They are already included in the alternative description. For any fixed collection of worldsheet maps, replacing the gerbe by its character preserves every background weight separately, rather than only a summed observable.

On a simply connected manifold one can instead specify smooth holonomy on based spheres modulo thin homotopy. The group \(\pi_2^2(X)\) consists of based spheres modulo thin homotopy, whose differential has rank at most two. Ref.~\cite{MackaayPicken} reconstructs the gerbe connection class from a smooth homomorphism \(\pi_2^2(X)\to\U(1)\): suitable double cones give loop holonomies on double intersections, ordinary line-bundle reconstruction gives their connections, and varying spheres recovers the local forms and transition functions. The ordinary group \(\pi_2(X)\) would not suffice. The volume example above varies under ordinary null-homotopies. The reconstruction identifies precisely the additional smooth holonomy data needed by this presentation.

\subsection{Loop bundles and fusion}

Surface sewing can also be performed with an ordinary principal bundle on the loop space \(LX=C^\infty(S^1,X)\). Take smooth paths that are constant near their endpoints, so that concatenation is smooth. On the space of triples \((\gamma_1,\gamma_2,\gamma_3)\) with common endpoints, define \(e_{ij}(\gamma_1,\gamma_2,\gamma_3)=\overline\gamma_j\star\gamma_i\in LX\) by following \(\gamma_i\) and returning along \(\gamma_j\). A fusion product on a principal \(\U(1)\)-bundle \(P\to LX\) is an isomorphism
\begin{align}
&\lambda:e_{12}^*P\otimes e_{23}^*P\longrightarrow e_{13}^*P
\label{eq:fusion}
\end{align}
that is associative for four paths. The connection preserves this product and obeys the symmetrizing and superficial conditions: the specified half-rotation respects fusion, rank-one tori have trivial holonomy, and rank-two homotopies preserve it. For connected \(X\), these bundles with connection are equivalent to {abelian} gerbes with connection after identifying two-isomorphic gerbe morphisms \cite{Waldorf2016}.

{A decomposition into cylinders, pairs of pants, and discs expresses a surface amplitude through parallel transport and the fusion product in Eq.~\eqref{eq:fusion}.} A cylinder is a path in \(LX\), so its amplitude is ordinary parallel transport in \(P\). A pair of pants combines two incoming fibres into the outgoing fibre using \(\lambda\). Cutting it into cylinders gives its amplitude by parallel transport and this single bundle map. Discs use the canonical section over constant loops. Sewing these pieces gives general surface amplitudes. Conversely, pulling \(P\) back to pairs of paths starting at a fixed base point and using \(\lambda\) reconstructs the {abelian} gerbe.

Thus transgression gives a complete alternative with specified operations. Its base is infinite-dimensional, and its fusion product is additional to the bare bundle. Whether it is a convenient simplification depends on the calculation: a single closed-surface phase may be easier to obtain from Eq.~\eqref{eq:worldsheet-character}, whereas repeated sewing can be expressed through Eq.~\eqref{eq:fusion}.

\subsection{Nonabelian transport by ordinary differential equations}

A concrete nonabelian transport construction begins with a principal \(G\)-bundle \(P\to X\), a connection \(A\), {and a crossed module as in Section~\ref{sec:ordinary-crossed-module}.} Let \(B\in\Omega^2(P,\fh)\) be horizontal and \(\alpha\)-equivariant, with \(t_*B=F_A\). Lift the bottom edge of a square \(\Sigma(s,t)\) horizontally, then lift its \(t\)-curves horizontally from that edge. For the resulting lift \(\widetilde\Sigma\), solve
\begin{align}
&\frac{dh}{ds}=(dR_{h(s)})_e
\int_0^1 B_{\widetilde\Sigma(s,t)}
(\partial_s\widetilde\Sigma,\partial_t\widetilde\Sigma)\,dt,
\qquad h(0)=1.
\label{eq:surface-ode}
\end{align}
Here \(R_h\) is right multiplication in \(H\). The surface transport is \(h(1)\). Ref.~\cite{LiHolonomy} derives the transformation under a change of initial lift and the composition laws using ordinary differential equations. Its nonabelian Stokes calculation reduces a three-dimensional variation to a two-dimensional curvature calculation under the stated thin-face assumptions.

This gives a direct procedure whenever the global \(P,A,B\) above are available. For general descent, local forms, overlap functions and 1-form gauge transformations are required. Ref.~\cite{SchreiberWaldorfHolonomy} gives their equivalence with complete smooth transport under its typical-fibre hypotheses. The adjusted version of these equations is written explicitly in Section~\ref{sec:data}. {Together, the differential equations and overlap data give local transport and global gluing without assuming a global trivialization.}

{
\subsection{Closed surfaces with a noncentral crossed module}
\label{sec:twisted-surface}

The differential character in Eq.~\eqref{eq:character} has scalar circle coefficients. For a general crossed module, the flat quotient connection transports the abelian kernel \(K\), and those coefficients must be transported along with the surface. Existing nonabelian differential cohomology already describes ordinary surface transport \cite{SchreiberWaldorfHolonomy,SatiSchreiber2026}. We give a concrete reduction of its descent equations to abelian differential cocycles with local coefficients. {A reference lift supplies the origin of this description, which is not determined by the quotient bundle alone.}

Fix a flat \(F\)-bundle \(P\to X\), including its flat connection, and an identification of the quotient of the 2-bundle with \(P\). The \(F\)-action on \(K\) defines a local system \(K_P^\delta\) of abelian groups and a flat vector bundle \(\fa_P\). The superscript \(\delta\) means locally constant coefficients. Write \(K_P^\infty\) for the sheaf of smooth sections of the associated \(K\)-bundle. {In a parallel frame, \(d\log\) is the pullback of the Maurer--Cartan form on \(K\), and \(d_\nabla\) is the ordinary exterior derivative. For the normalized circle, \(d\log z=(2\pi i)^{-1}z^{-1}dz\).}

\begin{proposition}
\label{prop:crossed-surface}
For the smoothly separable finite-dimensional crossed module of Section~\ref{sec:ordinary-crossed-module}, suppose that \(P\) has an ordinary fake-flat lift. Gauge-equivalence classes of lifts with the fixed quotient identification form a torsor for
\begin{align}
&\mathbb H^2(X;\mathcal D_P(2)),\qquad
\mathcal D_P(2)=
\bigl[K_P^\infty\xrightarrow{d\log}
\Omega^1(\fa_P)\xrightarrow{d_\nabla}\Omega^2(\fa_P)\bigr],
\label{eq:twisted-deligne}
\end{align}
with the sheaves in degrees \(0,1,2\). A reference lift identifies this torsor with the displayed abelian group. If \(X=\Sigma\) is a connected closed oriented surface, every flat \(P\) has a lift, and
\begin{align}
&\mathbb H^2(\Sigma;\mathcal D_P(2))
\cong H^2(\Sigma;K_P^\delta)
\cong H_0(\Sigma;K_P^\delta)
\cong K_{\rho(\pi_1\Sigma)}.
\label{eq:surface-coinvariants}
\end{align}
Here \(\rho:\pi_1\Sigma\to F\) is the monodromy. In additive notation, \(K_{\rho(\pi_1\Sigma)}\) is the quotient of \(K\) by the subgroup generated by \(\rho(\gamma)k-k\). The isomorphism to \(H_0\) is evaluation on the oriented surface with local coefficients.
\end{proposition}

\begin{proof}
Locally put the flat quotient connection in a parallel frame and lift that frame to \(G\). Then \(A_i\) takes values in \(\im t_*\). A 1-form transformation with \(t_*\varphi=A_i\) sets \(A_i=0\), and fake-flatness puts \(B_i\) in \(\fa\). This is the usual local abelianization of a fake-flat connection \cite{SchreiberWaldorfHolonomy,Waldorf2026}. On a contractible overlap, \(p(g_{ij})\) is constant. A transformation in \(H\) replaces \(g_{ij}\) by a constant representative of that coset. Equation~\eqref{eq:ordinary-two-connection} then puts the remaining \(\varphi_{ij}\) in \(\fa\).

Choose constant \(g_{ij}\) lifting the parallel quotient transitions and constant \(r_{ijk}\in H\) with \(t(r_{ijk})=g_{ik}(g_{jk}g_{ij})^{-1}\). Write \(h_{ijk}=k_{ijk}r_{ijk}\), \(k_{ijk}\in K\), and \(\rho_{ij}=\alpha_{g_{ij}}|_K\), also for its differentiated action. The descent equations become
\begin{align}
&B_j-\rho_{ij}B_i=-d\varphi_{ij},\qquad
\varphi_{jk}+\rho_{jk}\varphi_{ij}-\varphi_{ik}=d\log k_{ijk},
\label{eq:affine-surface-forms}\\
&k_{jkl}k_{ijl}k_{ikl}^{-1}(\rho_{kl}k_{ijk})^{-1}=c_{ijkl},\qquad
c_{ijkl}=
\bigl(r_{ikl}\alpha_{g_{kl}}r_{ijk}\bigr)
\bigl(r_{ijl}r_{jkl}\bigr)^{-1}\in K.
\label{eq:affine-surface-postnikov}
\end{align}
The constant twisted 3-cocycle \(c\) records the failure of the chosen fillers to satisfy quadruple coherence. In these parallel frames it represents the pulled-back Postnikov class of the underlying abstract crossed module \cite{BaezLauda2004}. Changing the fillers changes \(c\) by a twisted coboundary. Lift compatibility requires solving the equation for \(k\) in the sheaf of smooth \(K\)-sections. For continuous \(K\), this is a smooth-sheaf obstruction, even though the chosen cocycle \(c\) is constant.

Subtract a reference solution with these same quotient representatives and fillers. The relative fields \(\beta_i,\eta_{ij}\) and \(z_{ijk}=k_{ijk}/k^0_{ijk}\) satisfy
\begin{align}
&\beta_j-\rho_{ij}\beta_i=-d\eta_{ij},\qquad
\eta_{jk}+\rho_{jk}\eta_{ij}-\eta_{ik}=d\log z_{ijk},
\qquad \delta_\rho z=1.
\label{eq:relative-surface-deligne}
\end{align}
These are precisely the cocycles of Eq.~\eqref{eq:twisted-deligne}. Conversely, adding one to the reference gives a lift. A gauge transformation preserving the quotient has its local \(G\)-component in \(t(H)\). On each contractible patch, an \(H\)-transformation sets that component to the identity. The remaining 1-forms and overlap functions take values in \(\fa\) and \(K\), and change Eq.~\eqref{eq:relative-surface-deligne} by its abelian differential coboundaries. This proves both the classification and the torsor assertion.

On a surface choose a good cover with no fourfold intersections, for example open stars after a suitable barycentric subdivision. Constant \(g_{ij},r_{ijk}\), with \(B_i=\varphi_{ij}=0\) and \(k_{ijk}=1\), give a reference lift, since there is no quadruple equation. On a surface every 2-form has a local primitive, and every closed 1-form locally integrates to a logarithmic derivative. Thus \(\mathcal D_P(2)\) has only the locally constant \(K\)-sections as local cohomology and computes \(H^2(\Sigma;K_P^\delta)\). Poincar\'e duality with local coefficients identifies this with \(H_0(\Sigma;K_P^\delta)\). The latter identifies values transported around a loop and is exactly the coinvariant group in Eq.~\eqref{eq:surface-coinvariants}.
\end{proof}

{To compute the relative closed-surface class, pull the data and the reference back along \(f:\Sigma\to X\), and form Eq.~\eqref{eq:relative-surface-deligne}.} Choose \(\Lambda_i\in\Omega^1(U_i,\fa)\) with \(d\Lambda_i=\beta_i\). The overlap forms \(\eta_{ij}+\Lambda_j-\rho_{ij}\Lambda_i\) are closed, so choose functions \(v_{ij}\) with
\begin{align}
&dv_{ij}=\eta_{ij}+\Lambda_j-\rho_{ij}\Lambda_i,\qquad
\ell_{ijk}=z_{ijk}
\exp\bigl(-v_{jk}-\rho_{jk}v_{ij}+v_{ik}\bigr).
\label{eq:surface-constant-cocycle}
\end{align}
The logarithmic derivative of \(\ell\) vanishes by Eq.~\eqref{eq:relative-surface-deligne}, and \(\delta_\rho\ell=1\). Changing the primitives changes \([\ell]\) by a coboundary. Transport its coefficients to common fibres and evaluate on the oriented fundamental cycle. The result is the relative class in \(K_{\rho(\pi_1\Sigma)}\). This replaces a noncommuting transport calculation by local integration and an abelian cocycle evaluation, while keeping the flat quotient and reference data.

If \(K\) is a torus, let \(\Lambda=\ker(\exp:\fa\to K)\) and let \(\Lambda_P^\vee\) be the local system of dual integral lattices. Equation~\eqref{eq:twisted-deligne} is {degree-\(3\)} differential cohomology with lattice \(\Lambda_P\). Its character description uses weighted cycles:
\begin{align}
&\chi:Z_2(X;\Lambda_P^\vee)\longrightarrow\R/\mathbb Z,\qquad
\chi(\partial C)=\int_C\langle H_{\mathrm{rel}},\lambda\rangle\pmod{\mathbb Z}.
\label{eq:twisted-character}
\end{align}
Each simplex carries its parallel dual-lattice coefficient \(\lambda\), and \(H_{\mathrm{rel}}=d_\nabla\beta_i\) is the relative 3-curvature. {The usual character construction with these coefficients identifies Eq.~\eqref{eq:twisted-character} with the differential cocycles in Eq.~\eqref{eq:twisted-deligne}. Sign twists are treated explicitly in Ref.~\cite{GradySati2018}.} An unweighted surface need not have a global dual-lattice section, so its universal evaluation is the coinvariant class in Eq.~\eqref{eq:surface-coinvariants}. For discrete \(K\), the form terms vanish and the description is simply \(H^2(X;K_P^\delta)\). There is therefore no single scalar circle character for all crossed modules.

The quotient identification and the marking of a surface are separate choices. {Forgetting the quotient identification further identifies lift classes under automorphisms of \(P\), whose action relative to a reference can be affine.} A marked polygon also carries its handle transports. Its raw \(H\)-valued surface label lies over the boundary word under \(t\), and need not belong to \(K\) for positive genus. A reference with the same boundary transports makes the ratio \(K\)-valued. {The differential cocycles and those marking data reconstruct the transport. Equation~\eqref{eq:surface-coinvariants} classifies its fixed-quotient gauge-equivalence data.}

For an explicit noncentral example, take \(H=\U(1)\), \(G=\mathbb Z_2\), the trivial map \(t\), and the inversion action. On \(S^2\), a fixed quotient frame gives an arbitrary \(z\in\U(1)\). Forgetting the frame gives the orbit \(\{z,z^{-1}\}\), as in the orbit-valued sphere holonomy of Ref.~\cite{Parzygnat2015}. In particular, \(1\) and \(-1\) remain distinct. The reduction \(H/[G,H]\) loses this distinction: its commutators contain every square, so the quotient is trivial. On a surface with nontrivial inversion monodromy, however, the coinvariant group itself is trivial, because its relations quotient \(\R/\mathbb Z\) by all doubles. The fixed sphere frame, its automorphism orbit, and the coinvariants for a nontrivial flat background answer different questions.

Proposition~\ref{prop:crossed-surface} makes the closed-surface calculation explicit within existing crossed-module descent and differential cohomology. The reduction uses fake-flatness to choose parallel quotient frames. For an arbitrary adjusted connection the quotient may be curved, and the global normalization of Section~\ref{sec:global} addresses that case by a different change of fields.
}

\subsection{Monopole charge from ordinary bundle data}
\label{sec:monopoles}

An ordinary Dirac monopole has a global description by a line bundle on \(\R^3\setminus\{0\}\). Its curvature integral over a surrounding sphere gives the quantized charge. A higher surface-transport construction can also start from this ordinary information. Ref.~\cite{Parzygnat2015} begins with a principal connected Lie group bundle \(P\to X\), a connection, and a subgroup \(N\leq\pi_1(G)\). The corresponding covering homomorphism defines a crossed module,
\begin{align}
&\tau:\widetilde G_N\longrightarrow G,\qquad
\ker\tau=\pi_1(G)/N.
\label{eq:monopole-cover}
\end{align}
Since \(\tau_*\) is an isomorphism of Lie algebras, the local fake-flatness equation fixes the 2-form:
\begin{align}
&B=\tau_*^{-1}F_A.
\label{eq:monopole-curving}
\end{align}
There is no independent local 2-form to determine in this construction.

We can compute its sphere flux without constructing the associated 2-bundle. Choose a family of based loops sweeping the sphere. Their ordinary holonomies form a loop in \(G\), whose class modulo \(N\) is the higher flux in Ref.~\cite{Parzygnat2015}. Equivalently, trivialize \(P\) on the two hemispheres and lift the equatorial transition loop to \(\widetilde G_N\), starting at the identity. With orientation matched to the sweep, the endpoint is the same element of \(\ker\tau\). To see the equality, choose the hemisphere frames by parallel transport from their poles and identify the initial frames along one meridian. The equatorial transition is then the holonomy along a varying meridian followed by the fixed return meridian. This is exactly the loop of holonomies used in the surface construction.

Changing either hemisphere frame multiplies the transition loop by a loop extending over a disk, so its class is unchanged. The covering kernel is central and abelian, which also explains the gauge invariance and the abelian group of charge classes. {The charge here is determined by the ordinary bundle on the surrounding sphere. Comparing it with a monopole charge defined by a Higgs field would additionally require the corresponding bundle reduction.} For \(G=\U(1)\) and its universal cover, the kernel is the discrete group \(\mathbb Z\), and the construction records the original {degree-\(2\)} Chern class. It is not {an abelian \(\U(1)\)} gerbe with an independent {degree-\(3\)} characteristic class.

\subsection{Magnetic translations and the Jacobi obstruction}

Smooth magnetic charge gives a different problem. Ref.~\cite{BunkMuellerSzabo2019} starts on \(\R^d\) with a real 2-form \(\rho\) and \(H=d\rho\). The proposed position and momentum relations are
\begin{align}
&[x^i,p_j]=i\hbar\delta^i_j,\qquad
[p_i,p_j]=-i\hbar\rho_{ij}(x).
\label{eq:magnetic-commutators}
\end{align}
Using \([f(x),p_j]=i\hbar\partial_jf\), their cyclic left-nested commutator is
\begin{align}
&[[p_1,p_2],p_3]+[[p_2,p_3],p_1]+[[p_3,p_1],p_2]
=\hbar^2H_{123}(x).
\label{eq:magnetic-jacobi}
\end{align}
Ordinary associative operators on a common invariant domain satisfy the Jacobi identity, so they cannot realize these relations when the right-hand side is nonzero. Ref.~\cite{BunkMuellerSzabo2019} instead proposes a categorified state space built from a topologically trivial {abelian} gerbe with curving \(\rho\).

The numerical associator can nevertheless be computed directly. Using the ordered triangles and tetrahedra defined in Ref.~\cite{BunkMuellerSzabo2019}, the translation phases are
\begin{align}
\omega_{v,w}(x)=\exp\left(-\frac{i}{\hbar}\int_{\Delta^2(x;w,v)}\rho\right), \qquad
\omega_{u,v,w}(x)=\exp\left(\frac{i}{\hbar}\int_{\Delta^3(x;w,v,u)}H\right).
\end{align}
The second is the coboundary of the triangle phases, evaluated by Stokes' theorem. Thus its value and the Jacobi obstruction require only the forms and their integrals. This computation does not supply an ordinary Hilbert-space quantization of the same relations. Ref.~\cite{BunkMuellerSzabo2019} describes the proposed quantization as provisional and leaves the magnetic Weyl correspondence and Hamiltonian dynamics for further work.

An isolated Dirac charge avoids this obstruction by removing its singular position and using an ordinary line bundle on the complement. A self-dual string has a different transverse space, \(\R^4\setminus\{0\}\), and a 3-form flux through \(S^3\). Its abelian field can be described by the differential characters above \cite{RistSaemannWolf2026}. The synthetic monopole of Section~\ref{sec:bands} is instead a degeneracy in a Hamiltonian's parameter space.

\section{Wess--Zumino--Witten theory and string backgrounds}
\label{sec:string}

\subsection{Multiplication phases and Chern--Simons levels}

The Wess--Zumino--Witten model is a two-dimensional sigma model whose target is a Lie group. Its quantized Wess--Zumino term supports the current algebra that relates the model to three-dimensional Chern--Simons theory. For a compact, connected, simple target \(G\), take the invariant trace with the level normalization of Ref.~\cite{GawedzkiWaldorf2009}, and write
\begin{align}
&H_k=\frac{k}{24\pi^2}\tr(g^{-1}dg)^3,\qquad
\omega_k=\frac{k}{8\pi^2}\tr\bigl(g_1^{-1}dg_1\wedge g_2dg_2^{-1}\bigr).
\label{eq:wzw-forms}
\end{align}
Here \(H_k\) is a 3-form on \(G\), \(\omega_k\) is a 2-form on \(G\times G\), and integrality of \(H_k\) restricts the level \(k\). Choose a differential character with curvature \(H_k\), and let \(W_k(\varphi)\) be its surface holonomy on \(\varphi:\Sigma\to G\). For pointwise multiplication of two maps, the Polyakov--Wiegmann formula reads
\begin{align}
&W_k(\varphi_1\varphi_2)
=c_{\varphi_1,\varphi_2}W_k(\varphi_1)W_k(\varphi_2)
\exp\left(2\pi i\int_\Sigma(\varphi_1,\varphi_2)^*\omega_k\right).
\label{eq:polyakov-wiegmann}
\end{align}
The phase \(c\) is one for simply connected \(G\), but can be nontrivial otherwise. It measures a compatibility condition not determined by the curvature alone.

For \(G=SO(3)=\SU(2)/\mathbb Z_2\), this normalization allows a Wess--Zumino term at even \(k\). On a torus, label the winding classes of the two maps by \((m,n),(m',n')\in\mathbb Z_2^2\). The explicit phase is
\begin{align}
&c_{(m,n),(m',n')}=(-1)^{(k/2)(m'n-mn')}.
\label{eq:so3-phase}
\end{align}
Taking \((m,n)=(0,1)\) and \((m',n')=(1,0)\) gives \((-1)^{k/2}\). Hence the phase is identically one precisely when \(4\mid k\). Ref.~\cite{GawedzkiWaldorf2009} proves that this is the obstruction to the multiplicative structure used in its ordinary bosonic Chern--Simons construction. The stronger level restriction is therefore visible in the winding-number calculation. This statement does not exclude spin-dependent Chern--Simons refinements.

The condition here is that the specified phase be identically one, not merely that its class as an abstract group cocycle vanish. For \(k=2\bmod4\), the displayed sign is the coboundary of \((-1)^{mn}\) on \(\mathbb Z_2^2\), yet is not identically one. An arbitrary rephasing of that group cocycle does not supply the missing multiplicative structure for the fixed differential character. Ref.~\cite{GawedzkiWaldorf2009} develops this comparison through local forms, transition functions, multiplication cochains, and their associativity equations. These provide a complete description without the geometric gerbe model, while keeping the extra multiplication data. The phase calculation predates the multiplicative-gerbe interpretation \cite{FelderGawedzkiKupiainen1988,GawedzkiWaldorf2009}. {For a different construction based on differential crossed modules, a specified invariant polynomial makes the pure-gauge higher Wess--Zumino term vanish and the gauged term exact \cite{Song2026}. Those identities concern the associated higher Chern--Simons forms. They do not remove the global level obstruction for the ordinary Wess--Zumino--Witten model considered here.}

\subsection{Discrete torsion and projection operators}

Discrete torsion changes an orbifold spectrum through phases in the projection operators. Let \(\Gamma=\mathbb Z_2\times\mathbb Z_2\), writing its elements as \((a,b)\). The ordinary group cocycle
\begin{align}
&\omega((a,b),(a',b'))=(-1)^{ab'},\qquad
\epsilon(g,h)=\omega(g,h)\omega(h,g)^{-1}
\label{eq:discrete-torsion}
\end{align}
changes the projection in the \(g\)-twisted Hilbert space to
\begin{align}
&P_g^\epsilon=\frac1{|\Gamma|}\sum_{h\in\Gamma}\epsilon(g,h)U_h,
\label{eq:twisted-projector}
\end{align}
where \(U_h\) is the action of \(h\) in that sector. The cocycle identity ensures the consistency of these phases and their factorization.

For the \(T^6/\Gamma\) example of Ref.~\cite{VafaWitten1995}, Eq.~\eqref{eq:twisted-projector} transfers the forty-eight twisted-sector contributions from \(H^{1,1}\) to \(H^{2,1}\), giving
\begin{align}
&(h^{1,1},h^{2,1})=(51,3)\quad\longleftrightarrow\quad(3,51).
\label{eq:orbifold-spectrum}
\end{align}
Here \(h^{p,q}=\dim H^{p,q}\) are the Hodge numbers. The total displayed count is unchanged, but its grading changes. The two theories are mirror-related. This is the original calculation, expressed through phases and projections. A gerbe interpretation can describe the same twisting, but it supplies no missing step in this spectrum calculation.

Global compatibility also constrains which local resolutions can arise from a specified orbifold construction. For the \(T^7/(\mathbb Z_2)^3\) model of Ref.~\cite{BoyleSmithTachikawa2026}, the local phases must be restrictions of one class in the equivariant cohomology group \(H^2_{(\mathbb Z_2)^3}(T^7;\U(1))\). With complete desingularization, the eight remaining signs can be written as \((-1)^{\epsilon(a_3,a_5,a_7)}\), where \(a_3,a_5,a_7\in\mathbb F_2\) label the singularities and
\begin{align}
&\epsilon(a_3,a_5,a_7)=u+a_3v+a_5w,\qquad u,v,w\in\mathbb F_2.
\label{eq:g2-signs}
\end{align}
There are therefore zero, four, or eight zero values of \(\epsilon\). The resulting twisted-sector projections give only the Betti-number pairs \((b_2,b_3)=(8,47),(12,43),(16,39)\), among the nine geometric possibilities in the Joyce family. Ordinary equivariant cohomology gives the compatibility conditions and the projection calculation. This restriction concerns the stated orbifold construction. It does not exclude the remaining smooth \(G_2\) manifolds or their other possible worldsheet descriptions.

\subsection{A compact boson and decomposition}

A flat equivariant field can change a compact-boson spectrum through winding-dependent phases. In the cyclic orbifold of order \(N\) in Ref.~\cite{ChengMelnikovMinasian2022}, the generator acts on a state of winding number \(w\) by \(e^{2\pi i w/N}\). Projection keeps windings divisible by \(N\), while the sector indexed by \(b\in\mathbb Z/N\mathbb Z\) shifts the momentum number from \(n\) to \(n+b/N\). Here the unshifted momentum and winding numbers \(n,w\) are integers. If \(R\) is the original radius, these give momentum \((Nn+b)/(NR)\) and winding length \(NR(w/N)\). We thus recover the momentum and winding lattice of an ordinary circle of radius \(NR\) directly from the projection and shifted sectors.

More generally, Ref.~\cite{HellermanEtAl2006} proposes decomposition of noneffective finite-group orbifolds, in which a subgroup acts trivially on the target, into components indexed by representation orbits. The component twists and genus-dependent normalization factors are part of the proposed replacement. The general conformal-field-theory decomposition is stated as a conjecture. It is distinct from the explicit compact-boson calculation above.

\subsection{Open worldsheets and the boundary gauge field}

An open worldsheet also couples to a boundary gauge field on the D-brane. Its fermion Pfaffian and bulk surface factor need not separately be numbers. Their boundary transformation laws must cancel. Ref.~\cite{FreedWitten1999} expresses the cancellation through local connection forms and transition functions on the brane. The familiar characteristic-class condition involves the integral third Stiefel--Whitney class \(W_3\) of the brane normal bundle and the restricted integral bulk twist. Here \(W_3=\beta(w_2)\), where \(w_2\) is the second Stiefel--Whitney class and \(\beta\) is the integral Bockstein. The differential cancellation data, rather than the characteristic class alone, determine the amplitude.

The elementary example is a brane \(Q=\mathbb{CP}^2\) with topologically trivial bulk \(B\)-field. If \(x\) generates \(H^2(Q;\mathbb Z)\), then \(w_2(TQ)=x\bmod2\). A \(\mathrm{Spin}^c\) determinant line \(L\) satisfies \(c_1(L)=(2m+1)x\). The real boundary gauge potential is the corresponding half connection. Writing \(F_A\) for its real curvature,
\begin{align}
&\frac1{2\pi}\int_{\mathbb{CP}^1}F_A=m+\frac12.
\label{eq:freed-witten-example}
\end{align}
The half-integral flux follows from ordinary \(\mathrm{Spin}^c\) geometry and the determinant line. For nontrivial bulk twisting, the same calculation is made with twisted boundary transition functions and their compatible connection forms, following the cancellation formulas in Ref.~\cite{FreedWitten1999}. This example uses that concrete construction, without asserting a converse to every general anomaly-cancellation condition discussed there.

The boundary restrictions can also be computed in a Wess--Zumino--Witten model. For the level-two theory with target \(\mathrm{Spin}(8)/(\mathbb Z_2\times\mathbb Z_2)\), Ref.~\cite{Gawedzki2004} considers two {abelian} gerbes with the same curvature. On a particular symmetry-preserving conjugacy-class brane, one choice admits four rank-one boundary conditions. For the other, the Chan--Paton descent matrices satisfy
\begin{align}
&W_1^2=W_2^2=I,\qquad W_1W_2=-W_2W_1.
\label{eq:wzw-projective}
\end{align}
Rank one is impossible, and Pauli matrices give an irreducible rank-two solution. Every permitted rank is even. The flat choice therefore has a boundary consequence that ordinary projective representation theory computes. The noncommuting matrices are boundary gauge data for an abelian bulk gerbe.

\subsection{Green--Schwarz terms as differential cochains}

The Green--Schwarz mechanism cancels gauge and gravitational anomalies by coupling tensor fields to characteristic forms. In six dimensions its local input is a factorized anomaly polynomial \(I_8=\tfrac12Y\wedge Y\), where \(Y\) is a 4-form valued in the real span \(\Lambda_{\R}\) of the string-charge lattice \(\Lambda\), and the wedge product includes its integral bilinear pairing. Locally the tensor field strength obeys \(dH=Y\), and the action contains a term proportional to \(\tfrac12\int B\wedge Y\). For nontrivial bundles these forms require global completion.

A direct completion on spacetime \(M\) uses differential cochains, as in Ref.~\cite{MonnierMoore2019}. A degree-\(p\) cochain is a triple
\begin{align}
&\check c=(a,h,\omega),\qquad
d\check c=(\delta a,\omega-a-\delta h,d\omega),
\label{eq:differential-cochain}
\end{align}
where \(a\in C^p(M;\Lambda_{\R})\), \(h\in C^{p-1}(M;\Lambda_{\R})\), and \(\omega\in\Omega^p(M;\Lambda_{\R})\). A form in the cochain component means integration over smooth singular chains. Integral quantization, or a specified shifted quantization, is imposed on \(a\). The differential squares to zero. For cocycles, evaluation of \(h\) on cycles modulo \(\Lambda\) gives the holonomy. The cochain \(a\) includes torsion flux information absent from \(\omega\).

Cup products of these triples combine the ordinary cup product, the wedge product, and a cochain homotopy between them. The half-level Green--Schwarz term further requires a quadratic refinement and the appropriate Wu-class shift. Both enter the differential-cochain construction of Ref.~\cite{MonnierMoore2019}. Replacing them by the uncorrected expression \(\tfrac12 B\wedge Y\) would not preserve large-gauge invariance. The anomaly calculation requires the gravitational anomaly coefficient \(a_\mathrm{grav}\in\Lambda\) to be characteristic:
\begin{align}
&(a_\mathrm{grav},x)\equiv(x,x)\pmod2
\qquad\text{for every }x\in\Lambda.
\label{eq:characteristic-vector}
\end{align}
Here \((\cdot,\cdot)\) is the integral lattice pairing. Under the assumptions on the field model and string-defect anomaly inflow in Ref.~\cite{MonnierMoore2019}, global cancellation also involves a spin-bordism invariant.

{{On a spacetime with nontrivial topology, the global anomaly calculation requires differential cochains, characteristic classes, quadratic refinements, and bordism. An abelian gerbe can represent part of these data, while the differential-cochain formulation includes the compatibility conditions used in the anomaly calculation.} The replacement therefore preserves those conditions rather than reducing the problem to elementary differential forms.}

\subsection{Anomaly descent and 2-group symmetry in QED}

The same descent mechanism appears in ordinary four-dimensional quantum electrodynamics (QED). Consider \(N_f\geq2\) massless Dirac fermions of electric charge \(q\), with no dynamical magnetic monopoles, and use the simply connected flavor cover \(\SU(N_f)_L\times\SU(N_f)_R\). Write \(A_L,A_R\) for its background connections, \(\lambda_L,\lambda_R\) for their infinitesimal gauge parameters, and \(B\) for the 2-form background coupled to the magnetic current. In the antihermitian fundamental-trace convention of Ref.~\cite{CordovaDumitrescuIntriligator2019}, the background transformation is
\begin{align}
&\delta B=d\Lambda+\frac{q}{4\pi}
\bigl[\tr(\lambda_L\,dA_L)-\tr(\lambda_R\,dA_R)\bigr],
\label{eq:qed-background}
\end{align}
where \(\Lambda\) is a 1-form gauge parameter. This shift cancels the mixed operator-valued anomaly after the electric symmetry is gauged. Differentiating the background identity gives a mixed flavor-current Ward identity with a contact contribution proportional to the magnetic-current two-point function and to \(q\). Ordinary anomaly descent and the magnetic background coupling already compute this response. The 2-group description organizes their coupled transformation laws. Treating \(B\) as an independent abelian background would discard the coupling in Eq.~\eqref{eq:qed-background}.

\section{Band topology and tensor Berry connections}
\label{sec:bands}

\subsection{{The Dixmier--Douady number and the winding number}}

In band theory, the immediate tasks are to classify band structures with specified spectral gaps and to compute measurable responses. Introducing a tensor Berry connection can give a geometric interpretation of the same band invariant. For the three-band family in Ref.~\cite{PalumboGoldman2019}, the comparison is explicit:
\begin{align}
&H_{\mathrm{band}}(q)=
\begin{pmatrix}
0&0&d_1-id_2\\
0&0&d_3-id_4\\
d_1+id_2&d_3+id_4&0
\end{pmatrix},
\label{eq:band-model}
\end{align}
where \(M\in\mathbb R\) is a control parameter, \(q=(q_x,q_y,q_z)\in\mathbb T^3\) is crystal momentum in the Brillouin torus, \(d_1=\sin q_x\), \(d_2=\sin q_y\), \(d_3=\sin q_z\), and \(d_4=M-\cos q_x-\cos q_y-\cos q_z\). Assume \(|d(q)|>0\) everywhere. The negative-energy eigenstate is
\begin{align}
&|u_-\rangle=2^{-1/2}(u_1,u_2,-1)^{\mathsf T},
\quad u_1=(d_1-id_2)/|d|,\quad u_2=(d_3-id_4)/|d|.
\end{align}
The normalized vector \(n=d/|d|\), with \(d=(d_1,d_2,d_3,d_4)\), defines a map \(\mathbb T^3\to S^3\). Write \(\partial_x=\partial/\partial q_x\), and similarly for \(y,z\), and use \(\epsilon^{1234}=1\). Its degree is
\begin{align}
&w^3=\frac{1}{2\pi^2}\int_{\mathbb T^3}
\epsilon^{abcd}n_a\partial_xn_b\partial_yn_c\partial_zn_d\,d^3q.
\label{eq:band-winding}
\end{align}
Ref.~\cite{PalumboGoldman2019} constructs a local 2-form tensor connection \(B\) from these eigenstate components. {With \(c_i=\cos q_i\), direct expansion gives
\begin{align*}
&\epsilon^{abcd}n_a\partial_xn_b\partial_yn_c\partial_zn_d
=\frac{c_xc_y+c_yc_z+c_zc_x-Mc_xc_yc_z}{|d|^4}
=-\mathcal H_{xyz},
\end{align*}
where \(\mathcal H=dB\) has the curvature sign used in Ref.~\cite{PalumboGoldman2019}. Its Dixmier--Douady number is therefore}
\begin{align}
&\mathcal{DD}=\frac{1}{2\pi^2}\int_{\mathbb T^3}\mathcal H_{xyz}\,d^3q={-w^3}.
\label{eq:band-dd}
\end{align}
{The minus sign follows from the stated convention \(\epsilon^{1234}=1\). Reversing the target orientation used to define \(w^3\) makes the two numbers equal.} Consequently, the Dixmier--Douady number can be computed directly from the original Hamiltonian by Eq.~\eqref{eq:band-winding}. It does not distinguish additional members of this family beyond the winding number.

\subsection{{Measuring the charge through ordinary band geometry}}

This charge also has a direct experimental measurement. The tensor-monopole model proposed in Ref.~\cite{PalumboGoldman2018} was implemented using the three spin states of a nitrogen--vacancy centre in diamond \cite{ChenEtAl2022}:
\begin{align}
&\widehat H_{4D}=
\begin{pmatrix}
0&q_x-iq_y&0\\
q_x+iq_y&0&q_z+iq_w\\
0&q_z-iq_w&0
\end{pmatrix}.
\label{eq:nv-model}
\end{align}
At fixed \(H_0>0\), write \(q_x+iq_y=H_0\cos\alpha\,e^{i\beta}\) and \(q_z+iq_w=H_0\sin\alpha\,e^{i\phi}\), where \(0\leq\alpha\leq\pi/2\) and \(\beta,\phi\) are \(2\pi\)-periodic angles. This parametrizes a three-sphere surrounding the degeneracy at the origin.

For the negative-energy eigenstate \(|u_-\rangle\), define the ordinary Berry connection \(\mathcal A_\mu=i\langle u_-|\partial_\mu u_-\rangle\) \cite{Berry1984}, curvature \(\mathcal F_{\mu\nu}=\partial_\mu\mathcal A_\nu-\partial_\nu\mathcal A_\mu\), and Fubini--Study metric
\begin{align}
&g_{\mu\nu}=\operatorname{Re}\langle\partial_\mu u_-|
(1-|u_-\rangle\langle u_-|)|\partial_\nu u_-\rangle.
\end{align}
With the orientation of Ref.~\cite{ChenEtAl2022}, the curvature and metric of this Hamiltonian satisfy
\begin{align}
&\mathcal H_{\alpha\beta\phi}
=-\tfrac12(\mathcal F_{\alpha\beta}+\mathcal F_{\phi\alpha})
=4\sqrt{\det(g_{\mu\nu})},
\qquad \mu,\nu\in\{\alpha,\beta,\phi\},
\label{eq:nv-measurement}
\end{align}
so spectroscopic measurements of the curvature or metric give the integrated charge, equal to one in the ideal model. {The tensor-field interpretation describes the geometry of this engineered finite-level Hamiltonian, for which Eq.~\eqref{eq:nv-measurement} holds.} Under the chiral-symmetry-breaking perturbation studied in the supplementary material of Ref.~\cite{ChenEtAl2022}, the analogous quantity is not a topological invariant.

\subsection{A boundary Berry phase with zero curvature}

The preceding measurement uses curvature, but a flat Berry connection can still have nontrivial holonomy. Ref.~\cite{OhyamaTerashimaShiozaki2023} constructs a qutrit-chain family parametrized by \(L(3,1)\times S^1\), where the lens space \(L(3,1)=S^3/\mathbb Z_3\) is formed by scalar multiplication on \(S^3\subset\mathbb C^2\). Advancing around the circle returns the bulk Hamiltonian to its initial value while changing the boundary discrete Berry phase.

In the finite boundary models of Ref.~\cite{OhyamaTerashimaShiozaki2023}, the local Berry connections vanish. For a generator of \(\pi_1(L(3,1))=\mathbb Z_3\), the initial and final endpoint transition phases are \(1\) and \(\omega^2\), respectively, where \(\omega=e^{2\pi i/3}\). The ratio of the initial phase to the final phase is therefore \(1/\omega^2=\omega\). The semi-infinite calculation gives the same ratio after canceling the common bulk contribution. It does not require convergence of either separate semi-infinite phase. This obstructs a boundary Hamiltonian that is globally defined over the parameter space, periodic in the cycle, and uniquely gapped everywhere. Ordinary Berry holonomy, including its endpoint transition factor, reproduces the response even though the local curvature vanishes.

\subsection{Vanishing of higher wedge powers}

The six- and nine-dimensional extensions proposed in Ref.~\cite{PalumboGoldman2019} use wedge powers of a scalar abelian 3-form curvature \(\mathcal H\). Antisymmetry forces the square of this form to vanish:
\begin{align}
&\mathcal H\wedge\mathcal H
=(-1)^{3\cdot3}\mathcal H\wedge\mathcal H
=-\mathcal H\wedge\mathcal H,
\qquad
\mathcal H\wedge\mathcal H=0.
\label{eq:odd-square}
\end{align}
Over \(\R\) or \(\mathbb C\), the proposed \(\int\mathcal H^{\wedge2}\) and \(\int\mathcal H^{\wedge3}\) therefore vanish identically. They cannot supply the proposed nonzero higher-dimensional invariants as written. This does not affect Eq.~\eqref{eq:band-dd}, which integrates one 3-form. A product involving different fields or a different cohomology operation would be a different construction, not the displayed powers of one scalar curvature.

\subsection{Reconstruction from line integrals}

The reconstruction in Ref.~\cite{JankowskiSlagerPalumbo2025} defines a scalar by integrating a Berry-connection component along the shortest path from a base point and identifies its ordinary coordinate derivative with that component. Differentiation yields an additional curvature term. For a smooth 1-form \(A=A_i\,dk^i\) on a star-shaped coordinate patch, set \(x=k-k_0\) and
\begin{align}
&\phi(k)=\int_0^1x^jA_j(k_0+tx)\,dt.
\label{eq:radial-integral}
\end{align}

\begin{proposition}
\label{prop:radial-integral}
Let \(U\subset\mathbb R^n\) be star-shaped about \(k_0\), let \(A=A_i\,dk^i\in\Omega^1(U)\), set \(x=k-k_0\), and define \(\phi\) by Eq.~\eqref{eq:radial-integral}. Then, for every \(k\in U\),
\begin{align}
&\partial_i\phi(k)=A_i(k)
+\int_0^1t x^j(dA)_{ij}(k_0+tx)\,dt.
\label{eq:radial-correction}
\end{align}
Consequently, \(d\phi=A\) on \(U\) only if \(dA=0\) on \(U\).
\end{proposition}

\begin{proof}
Differentiating under the integral gives \(\int_0^1[A_i+t x^j\partial_iA_j]dt\). Integrating \(\frac{d}{dt}[tA_i(k_0+tx)]\) gives \(A_i(k)=\int_0^1[A_i+t x^j\partial_jA_i]dt\). Subtracting proves Eq.~\eqref{eq:radial-correction}. Applying \(d\) to \(d\phi=A\) proves the last assertion.
\end{proof}

For example, with \(k_0=0\), \(A=x\,dy\) gives \(\phi=xy/2\), not a scalar with \(d\phi=A\). A choice of singularities or multivalued branches cannot alter this local calculation on a smooth patch. For the proposed construction \(B=\phi\,dA\), the actual curvature is \(dB=d\phi\wedge dA\). Its identification with \(A\wedge dA\) does not follow from the stated line integral.

The latter equality fails even for smooth, gapped Bloch data. On a small ball about the origin with \(|x|<1/2\), take
\begin{align}
&u(x,y,z)=\begin{pmatrix}\sqrt{1/2-x}\\\sqrt{1/2+x}\,e^{i(y+z^2)}\end{pmatrix},
\qquad H_{\mathrm{loc}}=1-2uu^\dagger.
\label{eq:bloch-counterexample}
\end{align}
This normalized state has eigenvalue \(-1\), separated from the other eigenvalue \(+1\). Its Berry connection and radial integral are
\begin{align}
&\begin{aligned}
&A=iu^\dagger du=-(1/2+x)(dy+2z\,dz),\\
&\phi=-\left(y/2+xy/2+z^2/2+2xz^2/3\right).
\end{aligned}
\label{eq:bloch-radial}
\end{align}
Direct differentiation gives
\begin{align}
&A\wedge dA=0,\qquad
d(\phi\,dA)=\frac{xz}{3}\,dx\wedge dy\wedge dz.
\label{eq:bloch-curvatures}
\end{align}
Thus the proposed universal reconstruction does not yield its claimed Chern--Simons curvature. This conclusion concerns that construction, not independently established Hopf invariants or response formulas.

\section{T-duality and localization}
\label{sec:tduality}

\subsection{The ordinary Hori transform}

T-duality exchanges circle geometry and 3-form flux. Let \(\pi:Z\to X\) and \(\widehat\pi:\widehat Z\to X\) be a T-dual pair of principal circle bundles. Choose invariant connections \(A,\widehat A\), each integrating to one on its fibre, and closed invariant flux forms \(H,\widehat H\). On the correspondence space \(Z\times_X\widehat Z\), with projections \(p:Z\times_X\widehat Z\to Z\) and \(\widehat p:Z\times_X\widehat Z\to\widehat Z\), choose representatives satisfying
\begin{align}
&d(-p^*A\wedge\widehat p^*\widehat A)
=p^*H-\widehat p^*\widehat H.
\label{eq:hori-kernel}
\end{align}
Write \(d_H=d+H\wedge\). For a differential form \(\omega\) on \(Z\) invariant under its principal circle action, define the Hori transform by
\begin{align}
&T\omega=\widehat p_*\left(e^{-p^*A\wedge\widehat p^*\widehat A}p^*\omega\right).
\label{eq:hori}
\end{align}
Fibre integration is normalized by \(\widehat p_*(p^*A\wedge\widehat p^*\eta)=\eta\) for \(\eta\in\Omega(\widehat Z)\). The form identity in Eq.~\eqref{eq:hori-kernel} and integration over a one-dimensional fibre give \(d_{\widehat H}T=-Td_H\). To see that \(T\) is invertible, we use the unique decomposition \(\omega=\pi^*\omega_0+A\wedge\pi^*\omega_1\), with \(\omega_0,\omega_1\in\Omega(X)\). In this sign convention,
\begin{align}
&T\omega=\widehat\pi^*\omega_1-\widehat A\wedge\widehat\pi^*\omega_0.
\label{eq:hori-components}
\end{align}
This explicitly invertible operation reverses parity and induces the classical twisted-cohomology isomorphism \cite{Hori1999,BouwknegtEvslinMathai2004,HanMathai2026}. Its computation uses ordinary connections, forms, and fibre integration.

\subsection{Mapping-space localization and characteristic classes}

The loop-space construction of Ref.~\cite{HanMathai2015} uses the holonomy line bundle \(\mathcal L\to LZ\) with connection \(\nabla^{\mathcal L}\). Let \(\mathbb T=S^1\) act by rotation of loops, with infinitesimal generator \(K\). If \(\overline H\) is the averaged evaluation pullback of \(H\), the differential on rotation-invariant \(\mathcal L\)-valued forms is
\begin{align}
&D=\nabla^{\mathcal L}-u\iota_K+u^{-1}\overline H\wedge,
\qquad |u|=2.
\label{eq:loop-differential}
\end{align}
Here \(u\) is a formal variable of degree {\(2\)} and \(\iota_K\) denotes contraction with \(K\). The connection curvature cancels the contraction of \(\overline H\), so \(D^2=0\) on those forms. The complex is completed in the periodic grading: at a fixed total degree it permits arbitrarily high ordinary form degrees. Denote the cohomology of this completed complex by \(h^\bullet_{\mathbb T}(LZ,\nabla^{\mathcal L}:\overline H)\). Restriction to constant loops induces an isomorphism
\begin{align}
&\operatorname{res}_Z:
h^\bullet_{\mathbb T}(LZ,\nabla^{\mathcal L}:\overline H)
\xrightarrow{\ \cong\ }
H^\bullet\!\left(\Omega(Z)\widehat\otimes\mathbb C[u,u^{-1}],d+u^{-1}H\wedge\right).
\label{eq:loop-localization}
\end{align}
The hat indicates the preceding completion. The proof in Ref.~\cite{HanMathai2015} uses localization near the rotation-fixed maps and a contracting homotopy away from them. The localization theorem requires an invariant neighborhood of the fixed set with an invariant good cover and an equivariant deformation retraction. The loop-space example satisfies this condition near the constant loops.

Ref.~\cite{HanMathai2026} treats \(S^2Z=C^\infty(S^2,Z)\). After an axis of rotation is chosen, a fixed map is constant on each latitude and therefore determines a path from the south pole to the north pole. With normalized spherical measure, localization first restricts to these paths. Contraction to constant paths then includes an exponential correction for the twisting form. Write \(\operatorname{Loc}_Z\) for the induced isomorphism to the same target cohomology. In the T-duality construction the {abelian} gerbes with connection are additionally equivariant for the structure-circle actions, and the relevant forms are invariant under these actions.

To pass from ordinary twisted forms to the periodic complex, set \(c_u(\omega_j)=u^{-\lfloor j/2\rfloor}\omega_j\) for a homogeneous \(j\)-form and extend \(\mathbb C[u,u^{-1}]\)-linearly. This induces an isomorphism \(c_{u,H}\) on cohomology. The periodic Hori map is \(T_u=c_{u,\widehat H}Tc_{u,H}^{-1}\), and the sphere-space transform is defined by
\begin{align}
&T^{S^2}=\operatorname{Loc}_{\widehat Z}^{-1}\,T_u\,\operatorname{Loc}_Z.
\label{eq:transported-duality}
\end{align}
The inverse here is on cohomology. {Equation~\eqref{eq:transported-duality} transports the target-space Hori map through the localization isomorphisms and hence determines the sphere-space map from target-space data.}

The character maps below have as their common domain complex K-theory of \(Z\) with the fixed integral twist. Their targets are respectively ordinary twisted cohomology and the two mapping-space cohomologies.

\begin{proposition}
\label{prop:character-kernels}
Fix an integral twist on \(Z\), and assume the hypotheses under which the localization maps \(\operatorname{res}_Z\) and \(\operatorname{Loc}_Z\) above are isomorphisms. With the preceding periodic normalization, the three character maps on twisted complex K-theory satisfy
\begin{align}
&\ker\operatorname{BCh}_H=
\ker\operatorname{BBCh}_H=
\ker\operatorname{Ch}_H.
\label{eq:character-kernels}
\end{align}
\end{proposition}

\begin{proof}
The character maps satisfy
\begin{align}
&\operatorname{res}_Z\operatorname{BCh}_H=\operatorname{Ch}_H,
\qquad
\operatorname{Loc}_Z\operatorname{BBCh}_H=\operatorname{Ch}_H
\end{align}
\cite{HanMathai2015,HanMathai2026}. Both localization maps are isomorphisms. Their injectivity proves Eq.~\eqref{eq:character-kernels}.
\end{proof}

The mapping-space character classes consequently distinguish exactly the same twisted K-classes as the original character, and add no torsion information. Their explicit differential-form representatives and connection transgressions are additional geometric constructions before passage to cohomology. Thus localization gives a complete target-space computation of the classes and duality map, while comparison of their differential-form representatives remains a separate task. The integral twisting class in this comparison is specified independently of its curvature form.

\section{BF theory and particle--string linking}
\label{sec:bf}

\subsection{Compact fields and finite gauge theory}

BF theory couples a gauge curvature \(F\) to a multiplier field \(B\). In its compact abelian form it describes finite gauge theory and the topological phase obtained by linking a particle trajectory with a string worldsheet \cite{KapustinSeiberg2014}. We compute its partition function and particle--string linking amplitude.

Let \(Y\) be a closed oriented four-manifold, let \(N\) be a positive integer, and use compact fields with holonomies \(e^{i\oint A}\) and \(e^{i\int B}\). We take the Euclidean weight \(e^{-S_E}\), with local action
\begin{align}
&S_E[A,B]=\frac{iN}{2\pi}\int_Y B\wedge dA.
\label{eq:compact-bf}
\end{align}
Here \(A\) is a circle connection, while \(B\) is a 2-form gauge field with the transition data of Section~\ref{sec:surface}, rescaled by \(2\pi\). A global expression uses differential characters \(\widehat A\in\widehat H^2(Y;\mathbb Z)\), \(\widehat B\in\widehat H^3(Y;\mathbb Z)\):
\begin{align}
&e^{-S_E}=\exp\left(-2\pi iN\int_Y\widehat B\mathbin\cup\widehat A\right).
\label{eq:bf-differential}
\end{align}
The differential cup product has degree {\(5\)}, and integration over \(Y\) gives an element of \(\R/\mathbb Z\). This includes the torsion sectors that global differential forms alone would miss \cite{KapustinSeiberg2014}.

Dualizing the compact \(B\) field gives a scalar \(\phi\) of period \(2\pi\) with
\begin{align}
&d\phi+NA=0,\qquad
A\longmapsto A+d\alpha,\qquad
\phi\longmapsto\phi-N\alpha.
\label{eq:bf-scalar}
\end{align}
Locally setting \(\phi=0\) leaves the residual transformations \(\alpha=2\pi k/N\), giving an ordinary \(\mathbb Z_N\) gauge field. Compactness and the global field sum matter: the classical equations \(dA=dB=0\) by themselves would still allow continuous flat holonomies. With finite bundle-groupoid normalization, the result on connected \(Y\) is $Z_{\mathrm{BF}}(Y)=\frac{|H^1(Y;\mathbb Z_N)|}{N}$. Indeed, Fourier summation imposes \(\delta a=0\) on a finite 1-cochain \(a\), and dividing by 0-cochain gauge transformations gives \(|Z^1|/|C^0|=|H^1|/|H^0|\). On a connected closed spatial three-manifold \(X\), the resulting state count is \(\dim\mathcal H(X)=|H^1(X;\mathbb Z_N)|\). For \(X=T^3\), it is \(N^3\). The normalization is specified by the finite sum, rather than inferred from an unspecified continuum measure.

\subsection{The linking amplitude}

Take disjoint closed oriented submanifolds \(C\) and \(\Sigma\) of dimensions one and two in \(S^4\), with insertions
\begin{align}
&W_q(C)=e^{iq\oint_C A},\qquad
V_m(\Sigma)=e^{im\int_\Sigma B},\qquad q,m\in\mathbb Z_N.
\label{eq:bf-insertions}
\end{align}
Let the 2-current \(\delta_\Sigma\) satisfy \(\int_Y B\wedge\delta_\Sigma=\int_\Sigma B\). Integrating over \(B\) in the presence of \(V_m\) imposes \(dA=(2\pi m/N)\delta_\Sigma\). If \(D\) is an oriented 2-chain with \(\partial D=C\), Stokes' theorem gives the normalized correlation function
\begin{align}
&\langle W_q(C)V_m(\Sigma)\rangle
=\exp\left(\frac{2\pi iqm}{N}L(C,\Sigma)\right),\qquad
L(C,\Sigma)=D\cdot\Sigma.
\label{eq:bf-linking}
\end{align}
This is the particle--string Aharonov--Bohm phase, obtained by a field constraint and an intersection number \cite{KapustinSeiberg2014}. On general manifolds, the global cycle classes supply additional selection rules and torsion pairings. The displayed derivation uses \(S^4\).

On a manifold with boundary, \(B\mapsto B+d\Lambda\) changes the action by \((iN/2\pi)\int_{\partial Y}\Lambda\wedge dA\). Boundary conditions or compensating boundary fields must cancel this variation. Matching a closed-manifold sum does not permit us to discard it in a boundary theory.

\subsection{Euclidean and Poincar\'e BFCG actions}

The BFCG construction introduces multiplier fields for both curvatures of a higher connection \cite{GirelliPfeifferPopescu2008}. A direct comparison occurs for the Euclidean group \(\mathbb R^4\rtimes SO(4)\), or its Poincar\'e analogue. Decompose the ordinary connection into a rotational 1-form \(A\) and a translation-valued 1-form \(C\). Let \(B\) and \(\Sigma\) be 2-forms in the dual rotational and translation representations. Writing \(d_A\) for the covariant exterior derivative, the invariant pairings give
\begin{align}
&\begin{aligned}
&S_{\mathrm{BF}}=\int_Y\langle B\wedge F_A\rangle
+\int_Y\langle\Sigma\wedge d_A C\rangle,\\
&S_{\mathrm{BFCG}}=\int_Y\langle B\wedge F_A\rangle
+\int_Y\langle C\wedge d_A\Sigma\rangle.
\end{aligned}
\label{eq:bf-bfcg}
\end{align}
In the second action, \((A,\Sigma)\) are the higher connection fields, using the invariant metric to identify translation vectors and covectors. This is the BFCG formulation associated with the Euclidean or Poincar\'e 2-group \cite{GirelliTsimiklis2022}. The covariant Leibniz rule gives
\begin{align}
&S_{\mathrm{BF}}-S_{\mathrm{BFCG}}
=\int_Y d\langle\Sigma\wedge C\rangle
=\int_{\partial Y}\langle\Sigma\wedge C\rangle.
\label{eq:bf-boundary}
\end{align}
The sign follows from exchanging \(d_A\Sigma\), a 3-form, with \(C\), a 1-form. We assume the displayed fields and pairing define the global 3-form in Eq.~\eqref{eq:bf-boundary}.

On a closed manifold the two actions and their bulk equations agree. On a spatial hypersurface their presymplectic potentials differ by the field-space differential of \(\int\langle\Sigma\wedge C\rangle\), so the presymplectic 2-form is unchanged. Ref.~\cite{GirelliTsimiklis2022} interprets this as a partial change of polarization and studies its consequences for discretization. The higher connection therefore reorganizes these continuum fields without adding a bulk term to the action.

This identity also fixes the extent of the comparison with gravity. To obtain gravity from a BF-type theory, one must impose appropriate simplicity constraints and select a nondegenerate gravitational sector. Quantizing that sector is an additional problem \cite{DePietriFreidel1999}. Equation~\eqref{eq:bf-boundary} alone neither imposes those constraints nor identifies all constrained quantum amplitudes. More general crossed-module BFCG actions contain a fake-curvature coupling absent from Eq.~\eqref{eq:bf-bfcg} \cite{FariaMartinsMikovic2011}. Their reduction cannot be inferred from this particular boundary identity.

\section{Finite 2-group gauge theories and their actions}
\label{sec:finite}

Finite gauge theories describe topological information in gapped phases, and finite 2-group models have explicit lattice realizations with local commuting projectors \cite{KapustinThorngren2017,Huang2026}. We compare their closed partition functions, first by an exact Fourier transform and then by a quadratic example that shows where the ordinary-gauge reduction fails. The duality and its action restrictions are already recognized in Ref.~\cite{KapustinThorngren2017}. {More general invariants of flat 2-bundles can depend on an involutory Hopf algebra graded by the crossed module \cite{SozerVirelizier2026}. The sums below use cohomological actions and do not identify those additional Hopf-algebra invariants with ordinary finite-group gauge theories.}

\subsection{Three-dimensional partition functions}

Consider a finite homotopy gauge theory determined by a finite 2-group \(\mathcal G\). Its data can be written as finite groups \(Q,M\), with \(M\) abelian, an action \(\rho:Q\to\operatorname{Aut}(M)\), and a Postnikov class \([\beta]\in H^3(Q,M_\rho)\). For a crossed module these are its cokernel, kernel, induced action, and extension class. Choose a cocycle representative \(\beta\). A field on a triangulated manifold \(Y\) consists of a flat \(Q\)-bundle \(A\) and a 2-cochain \(B\in C^2(Y;M_A)\) satisfying \(\delta_A B=\beta(A)\), where \(M_A\) is the associated coefficient local system and \(\delta_A\) its coboundary. The field has the corresponding 1-gauge and 2-gauge transformations \cite{KapustinThorngren2017}.

Let \(\widehat M=\operatorname{Hom}(M,\U(1))\), with \(q\cdot\chi=\chi\circ\rho(q^{-1})\). For a flat bundle \(A\), write \(\widehat M_A\) for the associated dual local system. For a character \(\chi\), let \(Q_\chi\) be its stabilizer. The function
\begin{align}
&\omega_\chi(q_1,q_2,q_3)=\chi(\beta(q_1,q_2,q_3)),
\qquad q_i\in Q_\chi,
\label{eq:character-cocycle}
\end{align}
is an ordinary circle-valued 3-cocycle. Invariance of \(\chi\) under \(Q_\chi\) removes the coefficient action from the cocycle identity.

\begin{theorem}
\label{thm:finite-three}
Let \(\mathcal G\) be the finite 2-group with data \((Q,M,\rho,[\beta])\) specified above, and let \(Y\) be a closed connected oriented three-manifold. For the theory with zero topological action and homotopy-cardinality normalization \cite{MartinsPorter2007}, the partition function is
\begin{align}
&Z_{\mathcal G}(Y)=
\sum_{[\chi]\in\widehat M/Q}
Z_{\mathrm{DW}}(Y;Q_\chi,\omega_\chi).
\label{eq:finite-three}
\end{align}
On the right,
\begin{align}
&Z_{\mathrm{DW}}(Y;K,\omega)=
\sum_{[P]\in\operatorname{Bun}_K(Y)}
\frac{\langle P^*\omega,[Y]\rangle}{|\operatorname{Aut}P|}
\label{eq:dw-definition}
\end{align}
is the ordinary Dijkgraaf--Witten partition function. Here \(\operatorname{Bun}_K(Y)\) indexes isomorphism classes of principal \(K\)-bundles, and \(\operatorname{Aut}P\) is the gauge automorphism group of \(P\).
\end{theorem}

\begin{proof}
Fix a flat \(Q\)-bundle \(A\), its coefficient local system \(M_A\), and the obstruction \(\beta_A=A^*[\beta]\in H^3(Y,M_A)\). When \(\beta_A=0\), lifts form an \(H^2(Y,M_A)\)-torsor. Their 1-gauge and 2-gauge automorphisms contribute \(H^1(Y,M_A)\) and \(H^0(Y,M_A)\), respectively. The contribution over the marked bundle is therefore
\begin{align}
&\frac{|H^2(Y,M_A)|\,|H^0(Y,M_A)|}{|H^1(Y,M_A)|}
\mathbf1_{\beta_A=0}.
\end{align}
For finite cellular cochains, cancellation of kernels and images gives, with \(\chi(Y)\) denoting the Euler characteristic,
\begin{align}
&\prod_{j=0}^3|H^j(Y,M_A)|^{(-1)^j}=|M|^{\chi(Y)}=1.
\end{align}
Thus the contribution equals \(|H^3(Y,M_A)|\mathbf1_{\beta_A=0}\). Poincar\'e duality with local coefficients and character duality give a perfect pairing between \(H^3(Y,M_A)\) and \(H^0(Y,\widehat M_A)\). Character orthogonality yields
\begin{align}
&|H^3(Y,M_A)|\mathbf1_{\beta_A=0}
=\sum_{\chi\in H^0(Y,\widehat M_A)}\langle\chi,\beta_A\rangle.
\label{eq:finite-fourier}
\end{align}
Summing over \(A\) with weight \(1/|\operatorname{Aut}A|\) is the groupoid sum over pairs \((A,\chi)\). Since \(Y\) is connected, the parallel section \(\chi\) lies in one character orbit. A \(Q\)-bundle with a section in that orbit is equivalently a \(Q_\chi\)-bundle. Orbit--stabilizer gives the automorphism weight in Eq.~\eqref{eq:dw-definition}. Its phase is Eq.~\eqref{eq:character-cocycle}. This proves Eq.~\eqref{eq:finite-three}.
\end{proof}

This is the finite-character calculation behind the dualization in Ref.~\cite{KapustinThorngren2017}. The orbit/stabilizer expression also appears in the broader orbifold discussion of Ref.~\cite{PerezLonaSharpe2024}. If \(\rho\) is trivial and \(\beta=0\), it gives \(|M|\) copies of ordinary \(Q\)-gauge theory. For \(Q=\mathbb Z_2\), \(M=\mathbb Z_3\), inversion action, and \(\beta=0\), there are two character orbits and
\begin{align}
&Z_{\mathcal G}(Y)=Z_{\mathrm{DW}}(Y;\mathbb Z_2,1)+1.
\label{eq:inversion-example}
\end{align}
\subsection{Three-dimensional cohomological actions}

The decomposition is not limited to zero action. Write \(\widehat M=\operatorname{Hom}(M,\R/\mathbb Z)\) additively, with the same dual \(Q\)-action. The general bosonic {degree-\(3\)} cohomological action can be represented by
\begin{align}
&\begin{aligned}
&\lambda\in Z^1(Q,\widehat M),\qquad
\omega\in C^3(Q,\R/\mathbb Z),\qquad
\delta\omega=\langle\lambda\smile\beta\rangle,\\
&L(A,B)=\langle\lambda(A)\smile B\rangle+\omega(A).
\end{aligned}
\label{eq:finite-three-action}
\end{align}
The brackets pair dual coefficient systems, with transport in the cup product. The weight is \(\exp(2\pi i\langle L(A,B),[Y]\rangle)\). The Leibniz rule and \(\delta_A B=\beta(A)\) give \(\delta L=0\). Ref.~\cite{KapustinThorngren2017} derives this action and its dualization. In particular, there is no independent pure 2-form action in degree {\(3\)}. A four-dimensional quadratic term will behave differently.

For \(p\in\widehat M\), define an affine action and its stabilizer by
\begin{align}
&q\star p=q\cdot p+\lambda(q),\qquad
Q_p=\{q\in Q:q\star p=p\}.
\label{eq:affine-character-action}
\end{align}
The 1-cocycle identity for \(\lambda\) makes this a group action. On \(Q_p\), we have \(\delta p=-\lambda\), so
\begin{align}
&\omega_p=\omega|_{Q_p}+\langle p,\beta|_{Q_p}\rangle,
\qquad \delta\omega_p=0.
\label{eq:affine-cocycle}
\end{align}

\begin{proposition}
\label{prop:finite-action}
Let \((\lambda,\omega)\) satisfy Eq.~\eqref{eq:finite-three-action}, and let \(L\) be the resulting cohomological action. On every closed connected oriented three-manifold \(Y\), with the normalization of Theorem~\ref{thm:finite-three},
\begin{align}
&Z_{\mathcal G,L}(Y)=
\sum_{[p]\in\widehat M/(Q,\star)}
Z_{\mathrm{DW}}(Y;Q_p,e^{2\pi i\omega_p}).
\label{eq:finite-affine}
\end{align}
\end{proposition}

\begin{proof}
Fix \(A\). If its obstruction \(\beta_A\) vanishes, the lift classes form an \(H^2(Y,M_A)\)-torsor. The weighted sum over this torsor vanishes unless \([\lambda(A)]=0\), by its perfect pairing with \(H^1(Y,\widehat M_A)\). When \(\lambda(A)=-\delta_Ap\), integration by parts gives
\begin{align}
&\int_Y\bigl(\langle\lambda(A)\smile B\rangle+\omega(A)\bigr)
=\int_Y\bigl(\omega(A)+\langle p,\beta(A)\rangle\bigr).
\label{eq:finite-action-parts}
\end{align}
The solutions \(p\) form a torsor for \(H^0(Y,\widehat M_A)\), whose size is \(|H^3(Y,M_A)|\). This is exactly the lift factor calculated in Theorem~\ref{thm:finite-three}. If \(\beta_A\ne0\), summing the right-hand phase over the same torsor instead gives zero by character orthogonality. If \([\lambda(A)]\ne0\), that torsor is empty. Thus in all cases the fixed-\(A\) contribution is the sum over affine parallel sections with the phase in Eq.~\eqref{eq:finite-action-parts}. Summing over \(A\), a section in the orbit of \(p\) reduces its structure group to \(Q_p\). Orbit--stabilizer gives the ordinary bundle automorphism weight, and the remaining action is \(\omega_p\).
\end{proof}

The affine action changes the ordinary gauge groups as well as their phases. Equation~\eqref{eq:finite-affine} covers ordinary cohomological actions in \(H^3(B\mathcal G;\R/\mathbb Z)\), where \(B\mathcal G\) is the classifying space of the finite 2-group. Spin-dependent actions, extra matter, and prescribed boundary theories require their own data. This is a replacement of the stated field sum with its full action, not only a comparison of the untwisted number of fields.

\subsection{Four dimensions and the Euler factor}

{In four dimensions, the zero-action Fourier calculation gives the partition function of an ordinary 1-form gauge theory multiplied by an Euler factor.} Write the character group additively as \(\widehat M=\operatorname{Hom}(M,\R/\mathbb Z)\), and set
\begin{align}
&P=\widehat M\rtimes Q,\qquad
(\chi,q)(\chi',q')=(\chi+q\cdot\chi',qq').
\label{eq:semidirect}
\end{align}
Let \(\pi:P\to Q\) be projection and \(C(\chi,q)=\chi\), a 1-cocycle with coefficients in \(\widehat M\). Evaluation pairs \(\widehat M\) with \(M\). The circle-valued 4-cocycle is
\begin{align}
&\Omega=\exp\bigl(2\pi i\langle C\smile\pi^*\beta\rangle\bigr),
\quad
\Omega(p_1,p_2,p_3,p_4)
=\exp\bigl(2\pi i\chi_1(\rho(q_1)\beta(q_2,q_3,q_4))\bigr),
\label{eq:four-cocycle}
\end{align}
where \(p_i=(\chi_i,q_i)\). The coefficient transport in this cup product is necessary when \(\rho\) is nontrivial.

\begin{proposition}
\label{prop:finite-four}
Let \(\mathcal G\) be the finite 2-group with data \((Q,M,\rho,[\beta])\) specified above, and let \(P\) and \(\Omega\) be defined by Eqs.~\eqref{eq:semidirect} and \eqref{eq:four-cocycle}. For the zero-action theory on every closed connected oriented four-manifold \(Y\),
\begin{align}
&Z_{\mathcal G}(Y)=|M|^{\chi(Y)}Z_{\mathrm{DW}}(Y;P,\Omega).
\label{eq:finite-four}
\end{align}
\end{proposition}

\begin{proof}
The local-coefficient Euler identity expresses the fixed-\(A\) lift contribution as
\begin{align}
&|M|^{\chi(Y)}\frac{|H^3(Y,M_A)|}{|H^4(Y,M_A)|}\mathbf1_{\beta_A=0}.
\end{align}
Poincar\'e and character duality identify the duals of these cohomology groups with \(H^1(Y,\widehat M_A)\) and \(H^0(Y,\widehat M_A)\). Orthogonality makes the last expression
\begin{align}
&\frac{|M|^{\chi(Y)}}{|H^0(Y,\widehat M_A)|}
\sum_{c\in H^1(Y,\widehat M_A)}
\exp\bigl(2\pi i\langle c\smile\beta_A,[Y]\rangle\bigr).
\end{align}
The classes \(c\) describe lifts of the marked \(Q\)-bundle to \(P\)-bundles, with relative automorphism group \(H^0(Y,\widehat M_A)\). Summing over \(A\) produces Eq.~\eqref{eq:finite-four}. The cocycle identity for \(\Omega\) follows from the Leibniz rule, \(\delta C=0\), and \(\delta_\rho\beta=0\).
\end{proof}

On \(S^4\), the 2-group theory has partition function \(|M|/|Q|\), while ordinary \(P\)-gauge theory has \(1/(|M||Q|)\). The factor \(|M|^2\) in Eq.~\eqref{eq:finite-four} is consequently required. It can be represented by a local Euler term. Equations~\eqref{eq:finite-three}, \eqref{eq:finite-affine}, and \eqref{eq:finite-four} also preserve insertions depending only on the quotient bundle, after restriction or pullback in the ordinary theory. Boundary states and 2-form observables must themselves be Fourier-transformed. Equality of closed partition functions alone does not supply that dictionary.

{The zero-action hypothesis is stronger than the restriction required by the established four-dimensional duality. When the quadratic 2-form term vanishes, a linear term is specified by \(\lambda\in H^2(BQ;\widehat M)\), with the dual \(Q\)-action understood. Its compatibility condition is \(\langle\lambda\smile\beta\rangle=0\) in \(H^5(BQ;\R/\mathbb Z)\). A cochain trivialization supplies the required base counterterm. The Fourier sum then imposes \(\delta_A C=\lambda(A)\) on a dual 1-cochain \(C\). This is the flatness equation for an ordinary extension of \(Q\) by \(\widehat M\), determined by the action and \([\lambda]\). The extension need not split, so the semidirect product in Eq.~\eqref{eq:semidirect} applies specifically to \([\lambda]=0\). These statements, including the induced Dijkgraaf--Witten action, are the duality of Ref.~\cite{KapustinThorngren2017}.}

\subsection{Quadratic 2-form actions and the signature}

A quadratic term shows why the four-dimensional hypothesis matters. Set \(Q=1\), \(M=\mathbb Z_2\), and let \(b\in H^2(Y;\mathbb Z_2)\) be the 2-form class. The Pontryagin square is the cohomology operation
\begin{align}
&\mathcal P:H^2(Y;\mathbb Z_2)\longrightarrow H^4(Y;\mathbb Z_4),\qquad
\mathcal P(b+b')=\mathcal P(b)+\mathcal P(b')+2b\smile b'.
\label{eq:pontryagin-square}
\end{align}
The last term uses the inclusion \(\mathbb Z_2\to\mathbb Z_4\), \(1\mapsto2\). If \(b\) lifts to an integral class \(x\), then \(\mathcal P(b)=x\smile x\bmod4\). With the quadratic weight \(i^{\langle\mathcal P(b),[Y]\rangle}\), the partition function is
\begin{align}
&Z_q(Y)=\frac{|H^0(Y;\mathbb Z_2)|}{|H^1(Y;\mathbb Z_2)|}
\sum_{b\in H^2(Y;\mathbb Z_2)}i^{\langle\mathcal P(b),[Y]\rangle}.
\label{eq:quadratic-sum}
\end{align}
This is the finite 2-form action and normalization of Ref.~\cite{KapustinThorngren2013Theta}, specialized to the quadratic function \(q(1)=1/4\).

For positively oriented \(\mathbb{CP}^2\), the integral generator has square one. Its two mod-two classes contribute \(1\) and \(i\), and the gauge factor is two. On \(S^4\) only the zero class occurs. Consequently,
\begin{align}
&Z_q(\mathbb{CP}^2)=2(1+i),\qquad Z_q(S^4)=2.
\label{eq:quadratic-example}
\end{align}
The nondegenerate Gauss-sum formula in Ref.~\cite{KapustinThorngren2013Theta} evaluates the same theory on a closed oriented four-manifold as
\begin{align}
&Z_q(Y)=2^{\chi(Y)/2}\exp\left(\frac{2\pi i}{8}\sigma(Y)\right),
\label{eq:signature-response}
\end{align}
where \(\sigma(Y)\) is the signature of its intersection form. Removing the Euler factor leaves an invertible signature response.

This response cannot be absorbed into the ordinary finite-group gauge theory of Eq.~\eqref{eq:finite-four} with only an Euler normalization. On a simply connected manifold an ordinary finite-group gauge sum has only the trivial bundle, with weight \(1/|K|\). Reversing orientation does not change it. In contrast, the phase in Eq.~\eqref{eq:signature-response} changes from \(e^{i\pi/4}\) to \(e^{-i\pi/4}\) on \(\mathbb{CP}^2\), while its Euler characteristic stays fixed. An independent gravitational phase would have to be supplied. {Ordinary cochains, their Pontryagin square, and the finite Gauss sum nevertheless evaluate Eq.~\eqref{eq:quadratic-sum} directly. This restriction on ordinary finite-group duality was identified in Ref.~\cite{KapustinThorngren2017}.}

\section{Principal 2-bundles with adjusted connections}
\label{sec:data}

We now turn from fixed-background amplitudes and finite sums to the fields of a smooth nonabelian principal 2-bundle. {We first state the adjusted connection and gluing data for a central crossed module. For any adjustment satisfying Eq.~\eqref{eq:adjustment}, the 1-form gauge symmetry gives a local reduction, and a smooth bundle splitting gives a global normalization over a fixed quotient connection. An equivariant splitting makes the adjustment adapted and extends the global decomposition to the full gauge symmetry. We prove the adapted case first, then give the normalization and the compact and semisimple splitting results.}

\subsection{Connection and gluing data}
\subsubsection{The crossed module and its splitting}

{{We use} the finite-dimensional smoothly separable crossed module of Section~\ref{sec:ordinary-crossed-module}. For the adjusted reductions in this section we additionally assume centrality, meaning that the \(G\)-action on \(K\) is trivial. This imposes no commutativity assumption on \(G\) or \(H\).}

Write \(\fa,\fh,\fg,\ff\) for the Lie algebras of \(K,H,G,F\), respectively, and \(p:G\to F\) for the quotient map. The differentiated sequence and a splitting are
\begin{align}
&0\longrightarrow\fa\xrightarrow{i}\fh\xrightarrow{t_*}\fg
\xrightarrow{p_*}\ff\longrightarrow0,
\qquad
u:\fg\longrightarrow\fh,
\label{eq:sequence}
\end{align}
where \(t_*ut_*=t_*\) and \(ut_*u=u\). Put \(V=\im u\). The associated linear maps \(q:\ff\to\fg\) and \(j:\fh\to\fa\) satisfy
\begin{align}
&t_*u+qp_*=1,\qquad ut_*+ij=1,\qquad uq=ju=0.
\label{eq:splitting}
\end{align}
Thus \(\fh=i\fa\oplus V\), and \(t_*:V\to\im t_*\) is an isomorphism. The complements are vector spaces. The map \(q\) is not assumed to preserve Lie brackets. All these linear maps are fixed, independently of the spacetime point.

\subsubsection{The adjustment and the fields}

We use the adjustment conventions of Ref.~\cite{Waldorf2026}. An adjustment is a smooth map \(\kappa:G\times\fg\to\fh\), linear in its second argument, satisfying
\begin{align}
&\begin{aligned}
&\kappa(g_1g_2,X)
 =\kappa(g_1,\Ad_{g_2}X)+\kappa(g_2,X),\\
&\kappa(t(h),X)=(\widetilde\alpha_{h^{-1}})_*X,\\
&\kappa(g,t_*Y)=\alpha_gY-Y.
\end{aligned}
\label{eq:adjustment}
\end{align}
Here \(\alpha_g\) also denotes the induced action on \(\fh\), and \(\widetilde\alpha_h:G\to H\) is the map \(g\mapsto h^{-1}\alpha_g(h)\). {The adjustment is called \emph{adapted to \(u\)} if}
\begin{align}
&t_*\kappa(g,X)=t_*u(\Ad_gX-X).
\label{eq:adapted-adjustment}
\end{align}
The existence of a splitting alone does not assert the existence of such an adjustment. {The first global decomposition and the kinetic action use adaptedness.}

Let \(\kappa_*:\fg\times\fg\to\fh\) be the derivative of \(\kappa\) in its first argument at the identity, and let \(\alpha_*:\fg\times\fh\to\fh\) be the differentiated action. We will use
\begin{align}
&\kappa(1,X)=0,\qquad
\kappa_*(X,t_*Y)=\alpha_*(X,Y),\qquad
t_*\kappa_*(X,Y)=t_*u[X,Y].
\label{eq:infinitesimal-adjustment}
\end{align}
The first two follow from Eq.~\eqref{eq:adjustment}, and {the last holds for an adapted adjustment, by differentiating Eq.~\eqref{eq:adapted-adjustment}.}

Let \(M\) be a smooth spacetime manifold. On a trivializing open set \(U\subset M\), a connection is a pair \(A\in\Omega^1(U,\fg)\), \(B\in\Omega^2(U,\fh)\). Define
\begin{align}
&\begin{aligned}
&F_A=dA+\tfrac12[A\wedge A],\qquad \F=F_A-t_*B,\\
&a=p_*A,\qquad F_a=da+\tfrac12[a\wedge a]=p_*\F,\\
&\Hk=dB+\alpha_*(A\wedge B)+\kappa_*(A\wedge\F).
\end{aligned}
\label{eq:curvatures}
\end{align}
The bilinear operations in these formulas include the exterior product. The connection is \emph{fake-flat} if \(\F=0\), and \emph{adapted} if \(u\F=0\). An adapted adjustment is part of the fixed gauge law. An adapted connection is a restricted field configuration. We allow connections that are not adapted. {For an arbitrary adjustment, the condition \(u\F=0\) need not be preserved by the full gauge symmetry.}

\subsubsection{Gauge transformations and overlaps}

Let \(\bar\theta\) be the right Maurer--Cartan form on \(G\). An adjusted gauge transformation \((g,\varphi):(A,B)\to(A',B')\) consists of a map \(g:U\to G\) and \(\varphi\in\Omega^1(U,\fh)\) obeying
\begin{align}
&A'+t_*\varphi=\Ad_gA-g^*\bar\theta, \label{eq:gauge-A}\\
&d\varphi+\tfrac12[\varphi\wedge\varphi]
+\alpha_*(A'\wedge\varphi)+B'-\alpha_gB
=\kappa(g,\F).
\label{eq:gauge-B}
\end{align}
Composition is
\begin{align}
&(g_2,\varphi_2)\circ(g_1,\varphi_1)
=(g_2g_1,\varphi_2+\alpha_{g_2}\varphi_1).
\label{eq:composition}
\end{align}
The resulting curvature laws are
\begin{align}
&\F'=\Ad_g\F-t_*\kappa(g,\F),
\qquad
\Hk'=\Hk.
\label{eq:gauge-curvature}
\end{align}
These are the adjusted transformation identities of Ref.~\cite{Waldorf2026}.

Two gauge transformations with the same source and target may themselves be related by a map \(h:U\to H\). If \(\theta_H\) is the left Maurer--Cartan form on \(H\), its equations are
\begin{align}
&g_2=t(h)g_1,\qquad
\Ad_h^{-1}\varphi_2+(\widetilde\alpha_h)_*A'
=\varphi_1-h^*\theta_H.
\label{eq:gauge-two}
\end{align}
These higher transformations compare different choices of gluing on a triple overlap. Their equations contain \(A'\), but not \(B\).

We describe a principal 2-bundle with adjusted connection by the {nonabelian} bundle-gerbe descent data of Ref.~\cite{Waldorf2026}: fields \((A_i,B_i)\) on an open cover, adjusted transformations \((g_{ij},\varphi_{ij}):(A_i,B_i)\to(A_j,B_j)\) on double overlaps, and maps \(h_{ijk}\) on triple overlaps. {The triple and quadruple relations remain Eqs.~\eqref{eq:triple} and \eqref{eq:quadruple}. Neither involves \(B_i\).} Changing charts uses the same gauge transformations and transformations between them, with compatibility under composition. {Write \(\Gadj(M)\) for these connection and gluing data together with their transformations. When the adjustment is adapted, \(\Gapt(M)\) consists of the adapted connections and all transformations between them.}

{
\subsubsection{Fake curvature for an arbitrary adjustment}

The corrected action on fake curvature is part of the adjusted gauge laws \cite{RistSaemannWolf2026,Waldorf2026}. We record how it factors through the quotient and how it acts on the complement chosen by \(q\).

\begin{proposition}
\label{prop:curvature-representation}
For any adjustment satisfying Eq.~\eqref{eq:adjustment}, the maps
\begin{align}
&R_gX=\Ad_gX-t_*\kappa(g,X)
\label{eq:R-representation}
\end{align}
define a linear representation of \(F=G/t(H)\) on \(\fg\). This representation is the identity on \(\im t_*\) and induces the adjoint representation on \(\ff\). Fake curvature is a 2-form in the vector bundle associated with the ordinary quotient bundle by \(R\). For the fixed splitting,
\begin{align}
&\F=qF_a+t_*c,\qquad c=u\F,\qquad
c'=c+uR_gq(F_a).
\label{eq:affine-c}
\end{align}
The adjustment is adapted to \(u\) exactly when \(R_gq=q\Ad_{p(g)}\), in which case \(c'=c\).
\end{proposition}

\begin{proof}
Equivariance of \(t_*\) and the first and third adjustment axioms give
\begin{align}
R_{g_1}R_{g_2}X
&=\Ad_{g_1g_2}X-t_*\alpha_{g_1}\kappa(g_2,X)
-t_*\kappa(g_1,\Ad_{g_2}X)
+t_*\kappa(g_1,t_*\kappa(g_2,X))\notag\\
&=\Ad_{g_1g_2}X-t_*\kappa(g_2,X)
-t_*\kappa(g_1,\Ad_{g_2}X)=R_{g_1g_2}X.
\end{align}
Also \(R_1=1\) and \(R_g(t_*Y)=t_*Y\). Differentiating \(t(h\,\alpha_g(h^{-1}))=t(h)g\,t(h)^{-1}g^{-1}\) gives
\begin{align}
&t_*(\widetilde\alpha_{h^{-1}})_*X=\Ad_{t(h)}X-X.
\end{align}
The second adjustment axiom therefore implies \(R_{t(h)}=1\), so \(R\) factors through \(F\). Applying \(p_*\) gives its quotient action \(\Ad_{p(g)}\). {Equation~\eqref{eq:gauge-curvature} is \(\F'=R_g\F\), and on triple overlaps the transition maps compose because \(R_{t(h)}=1\).}

The splitting gives \(\F=qF_a+t_*u\F\). Applying \(uR_g\) proves Eq.~\eqref{eq:affine-c}, since \(ut_*\) is the identity on \(V\). If the adjustment is adapted, then
\begin{align}
&R_g=q\Ad_{p(g)}p_*+t_*u,
\label{eq:R-adapted}
\end{align}
which makes \(R_gq=q\Ad_{p(g)}\). {Conversely, that equality together with \(R_gt_*=t_*\) proves Eq.~\eqref{eq:R-adapted}. Subtracting it from \(\Ad_g\) recovers Eq.~\eqref{eq:adapted-adjustment}.}
\end{proof}

Thus for a fixed quotient connection \(c\) transforms affinely, with the translation determined by its ordinary curvature \(F_a\). The representation and affine law follow from the adjustment identities \cite{RistSaemannWolf2026,Waldorf2026}, without assuming adaptedness.

The same representation expresses the Bianchi identity in the quotient bundle. Let \(d_a^R\) denote the covariant exterior derivative on the vector bundle associated through \(R\). For \(X,Y\in\fg\), differentiating Eq.~\eqref{eq:R-representation} gives
\begin{align}
&R_*(p_*X)Y=[X,Y]-t_*\kappa_*(X,Y).
\end{align}
The ordinary Bianchi identity and the equivariance of \(t_*\) give
\begin{align}
&d\F+[A\wedge\F]
=-t_*\bigl(dB+\alpha_*(A\wedge B)\bigr),\qquad
t_*\Hk=-d_a^R\F.
\label{eq:adjusted-bianchi-image}
\end{align}
For an adapted adjustment, \(R_gq=q\Ad_{p(g)}\), so \(d_a^R(qF_a)=q(d_aF_a)=0\). {An adapted connection has \(\F=qF_a\), so its adjusted curvature is kernel-valued.}
}

\subsection{The ordinary global 2-form}
\label{sec:global}

{We first assume that the adjustment is adapted to \(u\). The selected part of fake curvature is then invariant under the full adjusted gauge symmetry and requires no nontrivial gauge gluing. After proving this decomposition, we replace adaptedness by a choice of bundle splitting over a fixed quotient connection.}

{Proposition~\ref{prop:curvature-representation} gives \(c'=c\) for every adjusted gauge transformation. Thus the local forms \(c_i=u(F_{A_i}-t_*B_i)\) agree on overlaps and define a global form \(c\in\Omega^2(M,V)\), unchanged by a change of 2-bundle presentation.}

For an ordinary global \(\fh\)-valued 2-form \(\lambda\), write \(\mathcal G^\lambda\) for the shift \(B_i\mapsto B_i+\lambda|_{U_i}\). This is the 2-form shift of Ref.~\cite{Waldorf2026}. Its effect on the left side of Eq.~\eqref{eq:gauge-B} is \(\lambda-\alpha_g\lambda=-\kappa(g,t_*\lambda)\), exactly matching the change of its right side. It therefore preserves adjusted transformations. Equation~\eqref{eq:gauge-two} and all their compositions are unchanged.

\begin{theorem}
\label{thm:global}
For the crossed module, splitting, and adapted adjustment specified in Section~\ref{sec:data}, and for every smooth manifold \(M\), {the connection and gluing data admit the following product description, compatible with all gauge transformations and natural under pullback:}
\begin{align}
&\Gadj(M)\simeq
\Gapt(M)\times\Omega^2(M,V)_{\disc}.
\label{eq:global-product}
\end{align}
The two maps on objects are
\begin{align}
&\mathcal G\longmapsto(\mathcal G^c,c),
\quad c=u\F,
\qquad
(\mathcal G_0,c)\longmapsto\mathcal G_0^{-c}.
\label{eq:global-maps}
\end{align}
The ordinary 2-form factor has only identity gauge arrows. If \(\mathcal H_0\) is the adjusted curvature of \(\mathcal G^c\), then
\begin{align}
&\Hk=\mathcal H_0-dc,\qquad
\mathcal H_0\in\Omega^3(M,i\fa).
\label{eq:global-curvature}
\end{align}
\end{theorem}

\begin{proof}
Since \(ut_*\) is the identity on \(V\), shifting \(B\) by \(c=u\F\) gives
\begin{align}
&u(\F-t_*c)=c-ut_*c=0.
\end{align}
Conversely, if \(u\F_0=0\), shifting an adapted connection by \(-c\) gives projected fake curvature \(c\). These constructions are inverse.

{Proposition~\ref{prop:curvature-representation}} makes the first construction compatible with every gauge arrow because its endpoints have the same \(c\). The shift by \(c\) preserves each gauge arrow and each transformation between arrows, with their original labels. Thus it preserves all composition and all descent equalities. The forms \(c_i\) glue uniquely, and the construction commutes with cover refinement. Pullback commutes with the fixed linear maps, curvature, and the shift \(B_i\mapsto B_i+c_i\), so both maps also commute with pullback. This proves Eq.~\eqref{eq:global-product} and its naturality. Distinct \(c\)'s have no connecting gauge arrow, which is the meaning of the subscript \(\disc\). It does not require \(c\) to be locally constant.

Finally, Eq.~\eqref{eq:infinitesimal-adjustment} gives
\begin{align}
&\Hk(A,B+\lambda)
=\Hk(A,B)+d\lambda
+\alpha_*(A\wedge\lambda)-\kappa_*(A\wedge t_*\lambda)
=\Hk(A,B)+d\lambda.
\end{align}
{Equation~\eqref{eq:adjusted-bianchi-image} places the curvature of the adapted connection \(\mathcal G^c\) in \(i\fa\).} Taking \(\lambda=c\) proves Eq.~\eqref{eq:global-curvature}.
\end{proof}

The theorem is an explicit combination of the invariance and shift identities in Ref.~\cite{Waldorf2026}. We have used no trivialization of the underlying 2-bundle or restriction on the topology of \(M\). Whenever an adapted object exists, every ordinary \(V\)-valued 2-form produces an unrestricted adjusted object with the same underlying topological 2-bundle. The shift to the adapted part is a change of variables when \(c\) is kept. It is not a gauge equivalence setting \(c\) to zero.

{
For an arbitrary adjustment, the affine law in Eq.~\eqref{eq:affine-c} prevents the local forms \(u\F\) from gluing in general. Fake curvature itself still glues through \(R\). A smooth splitting of that associated vector bundle therefore supplies a global normalization, even when no constant equivariant splitting exists.

\begin{theorem}
\label{thm:fixed-quotient-normalization}
Let \(\kappa\) be any adjustment in the central crossed-module class of Section~\ref{sec:data}. Fix an ordinary principal \(F\)-bundle \(P\to M\) with connection \(a\), and identify the quotient of each adjusted 2-bundle under consideration with \((P,a)\). The representation \(R\) defines a vector bundle \(E_P=P\times_R\fg\) and an exact sequence
\begin{align}
&0\longrightarrow M\times\im t_*
\longrightarrow E_P\xrightarrow{p_*}\operatorname{ad}P
\longrightarrow0.
\label{eq:associated-extension}
\end{align}
Here \(\operatorname{ad}P=P\times_{\Ad}\ff\). Choose a smooth bundle splitting \(q_P:\operatorname{ad}P\to E_P\), with \(p_*q_P=1\). Such a splitting always exists. For any adjusted lift \(\mathcal G\), put
\begin{align}
&c=(t_*|_V)^{-1}(\F-q_PF_a)\in\Omega^2(M,V).
\label{eq:bundle-normalization}
\end{align}
The maps
\begin{align}
&\mathcal G\longmapsto(\mathcal G^c,c),\qquad
(\mathcal G_0,c)\longmapsto\mathcal G_0^{-c},
\qquad \F_0=q_PF_a
\label{eq:bundle-normalization-maps}
\end{align}
are mutually inverse descriptions of the lift data and their gauge transformations inducing the identity on \((P,a)\). They are compatible with pullback when \(q_P\) is pulled back as well. This assertion decomposes the lifts that exist. It does not assert that every \((P,a)\) admits a lift.

Let \(d_a^R\) be the covariant exterior derivative on \(E_P\) induced by \(a\) and \(R\). If \(\mathcal H_0\) is the adjusted curvature of \(\mathcal G^c\), then
\begin{align}
&\Hk=\mathcal H_0-dc,\qquad
t_*\mathcal H_0=-d_a^R(q_PF_a)
=-(\nabla_aq_P)\wedge F_a.
\label{eq:bundle-normalized-curvature}
\end{align}
The derivative \(\nabla_aq_P\) uses the induced connections on \(\operatorname{ad}P\) and \(E_P\). In particular, a parallel splitting makes \(\mathcal H_0\) kernel-valued.
\end{theorem}

\begin{proof}
Proposition~\ref{prop:curvature-representation} gives the bundle \(E_P\), its quotient \(\operatorname{ad}P\), and the trivial subbundle in Eq.~\eqref{eq:associated-extension}: \(R\) fixes every vector in \(\im t_*\). Local splittings can be joined by a partition of unity because the condition \(p_*q_P=1\) is affine. Since \(p_*\F=F_a\), the difference \(\F-q_PF_a\) belongs to the trivial image subbundle. The inverse of \(t_*|_V\) therefore defines the ordinary global 2-form in Eq.~\eqref{eq:bundle-normalization}.

The 2-form shift used in Theorem~\ref{thm:global} needs only the third adjustment axiom. Thus \(B_i\mapsto B_i+c|_{U_i}\) preserves every overlap transformation and higher transformation with its original label, and gives \(\F_0=\F-t_*c=q_PF_a\). Shifting back by \(-c\) recovers the original data. A gauge transformation inducing the identity on the marked quotient acts trivially through \(R\), so its endpoints have the same \(c\). This proves compatibility with the allowed arrows. The constructions commute with refinement and with pullback of all the chosen data.

The shift identity is \(\Hk(A,B+c)=\Hk(A,B)+dc\), again using only the third adjustment axiom. Applying the second identity in Eq.~\eqref{eq:adjusted-bianchi-image} to \(\F_0=q_PF_a\), and using \(d_aF_a=0\), proves Eq.~\eqref{eq:bundle-normalized-curvature}.
\end{proof}

The quotient identification and the splitting specify which 2-form is separated. A quotient automorphism acts naturally on this description if it also transports \(q_P\). If \(q_P\) is kept fixed, the difference between the transported splitting and \(q_P\), applied to \(F_a\), gives an affine translation of \(c\). The fixed-quotient statement therefore does not imply a choice-independent product under the full quotient gauge group.

For a compact quotient, the bundle splitting can instead come from an equivariant linear splitting. This recovers the stronger natural product of Theorem~\ref{thm:global} for the original adjustment.

\begin{corollary}
\label{cor:compact-quotient}
If \(F\) is compact, every adjustment satisfying Eq.~\eqref{eq:adjustment} is adapted to some splitting with the same complement \(V\subset\fh\) and the same projection \(j\). Consequently Theorem~\ref{thm:global} applies without changing \(\kappa\), after choosing this splitting.
\end{corollary}

\begin{proof}
Let \(df\) be normalized Haar measure on \(F\). Average the fixed right inverse \(q\) by
\begin{align}
&\bar q=\int_F R_fq\Ad_{f^{-1}}\,df.
\label{eq:averaged-splitting}
\end{align}
Then \(p_*\bar q=1\), and invariance of Haar measure gives \(R_f\bar q=\bar q\Ad_f\). Define \(\ell=u(\bar q-q):\ff\to V\) and \(\bar u=u-\ell p_*\). Because \(\bar q-q\) takes values in \(\im t_*\), one has
\begin{align}
&t_*\ell=\bar q-q,\qquad
t_*\bar u+\bar q p_*=1,\qquad
\bar u\bar q=0,\qquad
\bar u t_*=ut_*.
\label{eq:averaged-u}
\end{align}
The image of \(\bar u\) is contained in \(V\), and \(\bar u t_*v=v\) for every \(v\in V\), so its image is exactly \(V\). The last equality in Eq.~\eqref{eq:averaged-u} keeps \(j\) unchanged. These identities define a splitting, and Proposition~\ref{prop:curvature-representation} makes the same adjustment \(\kappa\) adapted to \(\bar u\).
\end{proof}

The averaging is the usual splitting argument for finite-dimensional representations of a compact group. Its application here changes the splitting, whereas the projection of an infinitesimal adjustment to an adapted one in Ref.~\cite{LudewigWaldorf2025} changes the adjustment. With \(\bar q\), the associated splitting is parallel for every quotient connection, so the curvature coupling in Eq.~\eqref{eq:bundle-normalized-curvature} vanishes. The quadratic action in Section~\ref{sec:physics} can then be defined using the new splitting metrics, although a metric or action chosen before this change need not remain the same.

Compactness is sufficient for this construction, but the representation also admits a canonical splitting when the quotient Lie algebra is semisimple. In that case the quotient group may be noncompact or disconnected.

\begin{corollary}
\label{cor:semisimple-quotient}
If \(\ff\) is semisimple, there is a unique \(F\)-equivariant right inverse \(\bar q:\ff\to\fg\) of \(p_*\) for the representation \(R\). Every adjustment satisfying Eq.~\eqref{eq:adjustment} is therefore adapted to a splitting with the same \(V\) and \(j\), without changing \(\kappa\). The global product in Theorem~\ref{thm:global} and the kernel-valued normalized curvature follow for this splitting.
\end{corollary}

\begin{proof}
Differentiate the representation in Proposition~\ref{prop:curvature-representation}. The sequence \(0\to\im t_*\to\fg\xrightarrow{p_*}\ff\to0\) is then an exact sequence of finite-dimensional \(\ff\)-modules, with trivial action on \(\im t_*\) and adjoint action on \(\ff\). Complete reducibility for a semisimple Lie algebra gives an equivariant right inverse \cite{Knapp2002}. For a real Lie algebra, apply the complex theorem after complexification and average the resulting right inverse with its conjugate to obtain a real one.

If \(\bar q_1,\bar q_2\) are infinitesimally equivariant right inverses, their difference \(L:\ff\to\im t_*\) satisfies \(L([X,Y])=0\), since the target is a trivial module. Semisimplicity gives \([\ff,\ff]=\ff\) \cite{Knapp2002}, so \(L=0\). Infinitesimal equivariance implies equivariance under the identity component \(F^0\). For any \(f\in F\), the map \(R_f\bar q\Ad_{f^{-1}}\) is another \(F^0\)-equivariant right inverse, because \(F^0\) is normal in \(F\). Uniqueness therefore gives \(R_f\bar q=\bar q\Ad_f\) for all \(f\). Define \(\ell\) and \(\bar u\) as in the proof of Corollary~\ref{cor:compact-quotient}. Equation~\eqref{eq:averaged-u} and Proposition~\ref{prop:curvature-representation} give the claimed adapted splitting. Its associated bundle splitting is parallel, so Eq.~\eqref{eq:bundle-normalized-curvature} gives the curvature assertion.
\end{proof}

The semisimple argument uses complete reducibility and perfectness, not an invariant positive inner product. For a noncompact quotient, such an inner product need not exist, so the conclusion does not by itself supply the positive kinetic metrics used in Section~\ref{sec:physics}.

These results apply familiar splitting arguments to the representation carried by fake curvature while keeping the adjustment fixed. The global adjusted comparison in Ref.~\cite{BunkMuellerNuitenSzabo2026} takes a splitting and adapted adjustment as inputs. A recent linearization of higher gauge potentials at fixed curvature assumes positive weights on a finite-type Sullivan minimal Bianchi algebra \cite{Giotopoulos2026}. Here we fix the quotient connection and choose a splitting of the concrete vector-bundle sequence in Eq.~\eqref{eq:associated-extension}. The normalization then follows directly from the finite gauge and overlap equations.
}

\subsection{Local fields and classical actions}
\label{sec:local}

\subsubsection{An explicit normal form}

{{Throughout this subsection, we allow any adjustment satisfying Eq.~\eqref{eq:adjustment}, with centrality and the fixed splitting as above.} The 1-form gauge symmetry removes a component of \(A\). The field \(c\) remains independent and has the affine transformation law in Eq.~\eqref{eq:affine-c}.}

\begin{proposition}
\label{prop:normal}
Let \(U\subset M\) be a trivializing open set and let \((A,B)\in\Omega^1(U,\fg)\times\Omega^2(U,\fh)\) be an adjusted connection {for any adjustment satisfying Eq.~\eqref{eq:adjustment}.} Define
\begin{align}
&\begin{aligned}
&a=p_*A,\qquad x=uA,\qquad c=u(F_A-t_*B),\\
&b=jB-j\left(\alpha_*(qa\wedge x)+\tfrac12[x\wedge x]\right).
\end{aligned}
\label{eq:forward}
\end{align}
Here \(a\in\Omega^1(U,\ff)\), \(b\in\Omega^2(U,\fa)\), \(c\in\Omega^2(U,V)\), and \(x\in\Omega^1(U,V)\). The assignment \((A,B)\mapsto(a,b,c,x)\) is invertible, with inverse
\begin{align}
&\begin{aligned}
&A=qa+t_*x,\\
&B=ib+uF_{qa}-c+dx+\tfrac12[x\wedge x]+\alpha_*(qa\wedge x).
\end{aligned}
\label{eq:inverse}
\end{align}
Both substitutions have differential order at most one. The adjusted transformation \((1,x)\) takes the field to
\begin{align}
&A_0=qa,\qquad B_0=ib+uF_{qa}-c,\qquad
\F_0=qF_a+t_*c.
\label{eq:normal}
\end{align}
\end{proposition}

\begin{proof}
The infinitesimal crossed-module identities give
\begin{align*}
&F_{qa+t_*x}
=F_{qa}+t_*\left(dx+\alpha_*(qa\wedge x)+\tfrac12[x\wedge x]\right).
\end{align*}
Substitution into Eq.~\eqref{eq:inverse} yields \(\F=qF_a+t_*c\). Applying \(p_*,u,j\), and using Eq.~\eqref{eq:splitting}, recovers all four fields in Eq.~\eqref{eq:forward}. Conversely, these four projections determine \(A,B\), proving the inverse.

For \(g=1\), Eqs.~\eqref{eq:gauge-A}--\eqref{eq:gauge-B} become
\begin{align}
&A'=A-t_*\varphi,\qquad
B'=B-d\varphi-\tfrac12[\varphi\wedge\varphi]
-\alpha_*(A'\wedge\varphi).
\end{align}
Setting \(\varphi=x\) gives Eq.~\eqref{eq:normal}. Since \(\kappa(1,-)=0\), this transformation satisfies the adjusted gauge law. No condition \(c=0\) is needed.
\end{proof}

The substitutions are inverse on smooth fields, and their differentiated identities hold to every order. They are not assertions about equal-order finite-jet bundles: a derivative of \(c\) may use one more derivative of \(A\). They preserve locality because each substitution and its inverse use finitely many derivatives.

For \(\varphi\in\Omega^1(U,V)\), a pure 1-form transformation acts by
\begin{align}
&(a,b,c,x)\longmapsto(a,b,c,x-\varphi).
\label{eq:pure-shift}
\end{align}
The invariance of \(b\) follows by normalizing the transformed field: the composite parameters are \((x-\varphi)+\varphi=x\), giving the same normal representative.

\subsubsection{The full residual gauge structure}

{Normalization preserves all local gauge transformations and the transformations between them.} Let \(X,Y\) denote configurations and let \(N_X=(1,x_X)\) be their canonical normalizing transformations. An arrow \(T=(g,\varphi):X\to Y\) becomes
\begin{align}
&\overline T=N_Y\circ T\circ N_X^{-1}
=(g,x_Y+\varphi-\alpha_gx_X).
\label{eq:normalized-arrow}
\end{align}
A transformation \(h\) between arrows keeps the same label. Equation~\eqref{eq:composition} proves preservation of composition by cancellation of the intermediate \(x_Y\)-terms. The inverse formula is \(\varphi=\overline\varphi-x_Y+\alpha_gx_X\).

{Write \(\mathscr C(U)\) for the local fields \((A,B)\) together with their gauge transformations, and \(\mathscr C_0(U)\) for the normal fields \(x=0\) with all transformations between them.} Normalization gives the product presentation
\begin{align}
&\mathscr C(U)\cong
\mathscr C_0(U)\times\Pair\bigl(\Omega^1(U,V)\bigr).
\label{eq:local-product}
\end{align}
{Here \(\Pair\) means that there is exactly one shift arrow between every pair of \(x\)-values.} All other arrows and all transformations between them are in \(\mathscr C_0(U)\).

The remaining field transformations can be written directly, without the product presentation. Given \(g:U\to G\), put
\begin{align}
&a'=\Ad_{p(g)}a-p(g)^*\bar\theta_F,\qquad
\varphi_g=u\bigl(\Ad_g(qa)-g^*\bar\theta\bigr),
\label{eq:residual-parameters}
\end{align}
where \(p:G\to F\) is the quotient map and \(\bar\theta_F\) its right Maurer--Cartan form. A normalized arrow has \(\varphi=\varphi_g+i\lambda\), with \(\lambda\in\Omega^1(U,\fa)\), and
\begin{align}
&{c'=c+uR_gq(F_a)},\qquad b'=b+\omega_g(a)-d\lambda, \label{eq:residual-b}\\
&\omega_g(a)
=j\left(\alpha_g(uF_{qa})+\kappa(g,qF_a)
-\tfrac12[\varphi_g\wedge\varphi_g]
-\alpha_*(qa'\wedge\varphi_g)\right).
\label{eq:omega-g}
\end{align}
Indeed, \(j\,d\varphi_g=0\), the kernel is central, and the \(c\)-terms cancel through \(\kappa(g,t_*c)=\alpha_gc-c\). The transformations of \(a\) and \(b\) are therefore generally coupled. Equation~\eqref{eq:local-product} does not replace them by an independent product of ordinary gauge groups.

{
\begin{remark}
\label{rem:nonadaptable}
{The local reduction applies even when no splitting makes the adjustment adapted.} Take additive \(G=\R^2\), \(H=\R\times\U(1)\), \(t(v,z)=(v,0)\), and trivial action. In Lie-algebra coordinates set
\begin{align}
&\kappa((s,r),(X,Y))=(rY,0),\qquad
R_{(s,r)}(X,Y)=(X-rY,Y).
\label{eq:shear-adjustment}
\end{align}
The first adjustment axiom follows from additivity in \(r\). The second holds because \(r=0\) on \(t(H)\), and the third because \(Y=0\) on \(\im t_*\). With \(u(X,Y)=(X,0)\), the quotient is \(F=\R\), \(a=A_2\), and \(c=dA_1-B_1\). Hence \(c'=c-rF_a\). Since \(G\) is abelian, Eq.~\eqref{eq:adapted-adjustment} would require \(t_*\kappa=0\) for every splitting, which this adjustment violates. Proposition~\ref{prop:normal} and Theorem~\ref{thm:actions} still give the local reduction, and Theorem~\ref{thm:fixed-quotient-normalization} gives the global normalization once the quotient connection and bundle splitting are fixed.
\end{remark}
}

\subsubsection{Local observables, actions and equations}

A finite-jet local density depends smoothly on the fields and finitely many of their derivatives at a point. We write \(\dH\) for the total spacetime derivative of such local expressions. All variational statements below concern a field-coordinate neighborhood on which the path \(x\mapsto sx\), \(0\leq s\leq1\), is defined.

\begin{theorem}
\label{thm:actions}
Let \(U\) be a trivializing open set, and let \(L\) be a finite-jet local Lagrangian density for unrestricted adjusted connections on \(U\), {with any adjustment satisfying Eq.~\eqref{eq:adjustment}.} Suppose that the action \(\int_U L\) is invariant off shell under every compactly supported pure shift in Eq.~\eqref{eq:pure-shift}. In the coordinates of Proposition~\ref{prop:normal}, there exist a finite-jet local density \(L_{\mathrm{red}}(a,b,c)\) and a finite-jet local \((\dim U-1)\)-form \(K\) such that
\begin{align}
&L(a,b,c,x)=L_{\mathrm{red}}(a,b,c)+\dH K.
\label{eq:action-descent}
\end{align}
Every strictly shift-invariant local observable depends on no jet of \(x\). The substitution identifies the local Euler--Lagrange solution germs, modulo the \(x\)-shift, with those of \(L_{\mathrm{red}}\).
\end{theorem}

\begin{proof}
Arbitrary compactly supported shifts imply the Euler--Lagrange identity \(E_x(L)=0\) by the fundamental lemma of the calculus of variations. Applying the finite-jet first-variation formula to \(x_s=sx\) gives
\begin{align}
&\frac{d}{ds}L(a,b,c,x_s)
=E_x(L)\cdot x+\dH\Theta_L(a,b,c,x_s;x)
=\dH\Theta_L(a,b,c,x_s;x).
\end{align}
Here \(\Theta_L\) is the boundary term obtained by integration by parts. Integration over \(s\) proves Eq.~\eqref{eq:action-descent}, with \(L_{\mathrm{red}}=L|_{x=0}\). Integration by parts can increase the derivative order of \(K\), but it remains finite. Compact support was used to obtain the Euler identity. The actual field \(x\) need not have compact support. For a strictly invariant observable, arbitrary finite jets of the shift parameter directly imply independence of all \(x\)-jets.

To compare the equations before the change of fields, write \(z=(A,B)\), \(w=(a,b,c,x)\), \(z=\Psi(w)\), and \(w=\Phi(z)\). Linearizing the two inverse differential substitutions gives \(D\Phi\,D\Psi=1\) and \(D\Psi\,D\Phi=1\) at corresponding jets. The variational chain rule is
\begin{align}
&E_w(L\circ\Psi)=D\Psi^\dagger E_z(L),
\qquad
E_z(L)=D\Phi^\dagger E_w(L\circ\Psi),
\end{align}
where \(\dagger\) is the formal adjoint defined by integration by parts. These identities prove both implications between the equations. The total derivative has no bulk Euler--Lagrange contribution, and the remaining \(x\)-equation is identically zero. No inverse differential equation or Green operator is used.
\end{proof}

If the action has the full adjusted gauge symmetry, \(L_{\mathrm{red}}\) has the residual symmetry transported by Eq.~\eqref{eq:normalized-arrow}. {The remaining fields \(a,b,c\) can therefore interact through any local terms allowed by this symmetry. Section~\ref{sec:physics} gives explicit examples. Extending the comparison to arbitrary Batalin--Vilkovisky (BV) complexes or quantum generating functionals also requires the action-dependent gauge resolution and the treatment of sources and measures, as in the respective local-cohomology and field-redefinition results of Refs.~\cite{BarnichBrandtHenneaux1995,CriadoPerezVictoria2019}.}

\section{Dynamics in the reduced variables}
\label{sec:physics}

We now use the substitution to express the adjusted curvature and invariant actions in the reduced variables. {Throughout this section we again assume an adjustment adapted to \(u\), so that \(c\) is invariant and the adjusted curvature has the decomposition below. Corollaries~\ref{cor:compact-quotient} and \ref{cor:semisimple-quotient} supply such a splitting, while the positive metrics used for the kinetic action are an additional requirement. The shear representation in Eq.~\eqref{eq:shear-adjustment} has no positive definite invariant inner product, since all vectors \((-r,1)\) would have the same norm as \(r\) varies. Its fixed-quotient normalization therefore gives a wider class of global field descriptions than the kinetic construction used here.}

\subsection{Curvature and a quadratic action}
\label{sec:curvature-action}

In the normal form, the adjusted curvature separates as
\begin{align}
&\begin{aligned}
&\Hk=iH_b-dc,\qquad H_b=db+\Omega(a),\\
&\Omega(a)=j\alpha_*(qa\wedge uF_{qa})
+j\kappa_*(qa\wedge qF_a).
\end{aligned}
\label{eq:H-normal}
\end{align}
The \(c\)-terms cancel algebraically: the contribution \(-\alpha_*(qa\wedge c)\) from \(B_0\) cancels \(\kappa_*(qa\wedge t_*c)\). {With \(c=0\), fake curvature is \(qF_a\), and Eq.~\eqref{eq:adjusted-bianchi-image} places the remaining adjusted curvature in \(i\fa\). Its \(j\)-projection is \(db+\Omega(a)\), because \(ju=0\), proving Eq.~\eqref{eq:H-normal}. This uses the equivariance of \(q\) under \(R\), without requiring \(q\) to preserve Lie brackets.}

Choose positive inner products on \(\fa,V\) and a positive \(\Ad_F\)-invariant inner product on \(\ff\). The last assumption holds, for example, for compact \(F\). Define splitting metrics by
\begin{align}
&|X|_{\fg}^2=|p_*X|_{\ff}^2+|uX|_V^2,
\qquad
|Y|_{\fh}^2=|jY|_{\fa}^2+|ut_*Y|_V^2.
\label{eq:metrics}
\end{align}
The first metric need not be \(\Ad_G\)-invariant. Its norm on fake curvature is invariant under the actual adjusted gauge law, since \(p_*\F=F_a\) transforms under \(\Ad_F\), while \(u\F=c\) is invariant.

On an oriented Riemannian spacetime, let \(\|\cdot\|^2\) include the corresponding Hodge inner product and integration. For \(e,g_2>0\), consider
\begin{align}
&S_E[A,B]=\frac{1}{2e^2}\|\F\|_{\fg}^2
+\frac{1}{2g_2^2}\|\Hk\|_{\fh}^2.
\label{eq:quadratic-action}
\end{align}
Equations~\eqref{eq:normal} and \eqref{eq:H-normal} give the exact equality
\begin{align}
&S_E=
\frac{\|F_a\|_{\ff}^2}{2e^2}
+\frac{\|db+\Omega(a)\|_{\fa}^2}{2g_2^2}
+\frac{\|dc\|_V^2}{2g_2^2}
+\frac{\|c\|_V^2}{2e^2}.
\label{eq:action-split}
\end{align}
The field \(c/g_2\) is an ordinary massive 2-form of mass \(m=g_2/e\). In Lorentzian dimension \(D\geq3\), its free massive equation imposes \(d^\dagger c=0\), where \(d^\dagger\) is the formal adjoint of \(d\) for the spacetime metric. In a rest frame its independent components are the spatial \(c_{ij}\). Thus it has \(\dim V\binom{D-1}{2}\) polarizations, including \(10\dim V\) in six dimensions.

The decoupling of \(c\) in Eq.~\eqref{eq:action-split} is a property of this quadratic action. The global product in Theorem~\ref{thm:global} is only a decomposition of gauge data. It permits actions that couple its factors.

\subsection{An interaction with a nonabelian quotient}

{A finite-dimensional example couples \(c\) to a nonabelian gauge field.} Take
\begin{align}
&G=\SU(2)\times\R,\qquad H=\R\times\U(1),\qquad
t(v,z)=(1,v),
\label{eq:example-module}
\end{align}
with trivial action of \(G\) on \(H\). The image of \(t\) is central, \(H\) is abelian, and the crossed-module identities hold. Its kernel is \(\U(1)\) and its quotient is \(\SU(2)\). At Lie algebra level choose
\begin{align}
&u(X,v)=(v,0),\qquad qX=(X,0),\qquad j(v,z)=z,\qquad\kappa=0.
\end{align}
The zero adjustment is adapted, since \(u(\Ad_gX-X)=0\).

Writing \(A=(a,\alpha)\), \(B=(\beta,b)\), one has \(c=d\alpha-\beta\) and \(\Omega(a)=0\). The transformation with 1-form parameter \((\alpha,0)\) gives
\begin{align}
&A'=(a,0),\qquad B'=(-c,b),\qquad\Hk=(-dc,db).
\end{align}
The quadratic action therefore describes \(\SU(2)\) Yang--Mills theory, an ordinary abelian 2-form gauge field, and an ordinary massive real 2-form.

In six dimensions this system also admits
\begin{align}
&S_{\mathrm{int}}=\lambda\int c\wedge P_4(F_a),
\qquad
P_4(F_a)=-\tr(F_a\wedge F_a),
\label{eq:c-interaction}
\end{align}
where the trace is in the defining representation and \(\lambda\in\R\). Both \(c\) and \(P_4\) are gauge invariant. The field \(c\) is a real tensor, not a compact 2-form gauge potential, so large 2-form gauge transformations impose no quantization condition on \(\lambda\). In Euclidean signature its equation is
\begin{align}
&g_2^{-2}d^\dagger dc+e^{-2}c+\lambda *P_4(F_a)=0.
\label{eq:c-source}
\end{align}
Consequently, imposing \(c=0\) is not a consistent truncation on generic Yang--Mills configurations with \(P_4(F_a)\ne0\).

The cubic vertex is nonzero on physical polarizations. In six-dimensional Lorentzian signature \((-+++++)\), take two null covectors \(k_1=E(dt+dx^5)\), \(k_2=E(dt-dx^5)\), with \(E=m/2\), gauge polarizations \(dx^1,dx^2\) in a common nonzero Cartan generator \(T\), and massive tensor polarization \(dx^3\wedge dx^4\). All polarizations are transverse, and \((k_1+k_2)^2=-m^2\). The vertex contains
\begin{align*}
&-\tr(T^2)\,dx^3\wedge dx^4\wedge
(k_1\wedge dx^1)\wedge(k_2\wedge dx^2),
\end{align*}
a nonzero multiple of the volume form. Thus the invariant tensor can participate in ordinary interactions.

\subsection{Interactions with abelian gauge symmetry}
\label{sec:quartic}

An abelian 2-form can also have gauge-invariant curvature interactions. On six-dimensional Minkowski spacetime let \(B\) be a real, nonchiral 2-form, \(H=dB\), and
\begin{align}
&I=\frac1{3!}H_{\mu\nu\rho}H^{\mu\nu\rho},
\qquad
\mathcal L=-\frac12 I+\frac{\gamma}{4\Lambda^6}I^2,
\label{eq:quartic}
\end{align}
where \(\Lambda>0\) is a mass scale and \(\gamma\ne0\) is dimensionless. The gauge symmetry remains \(B\mapsto B+d\lambda\). The field dimensions are \([B]=2\), \([H]=3\), so the second term is an irrelevant effective-field-theory interaction, used perturbatively below \(\Lambda\). Its equation is
\begin{align}
&\partial_\mu\left[\left(1-\frac{\gamma}{\Lambda^6}I\right)
H^{\mu\nu\rho}\right]=0.
\label{eq:nonlinear-eom}
\end{align}

Nonlinearity alone would not exclude a free theory written in different variables. We therefore compute a four-point on-shell amplitude. Take all four external momenta in directions \(0,1,2,3\), with each polarization in the \(45\)-plane. These polarizations are transverse and not pure gauge, since \((k_i)_4=(k_i)_5=0\) for every external momentum \(k_i\). On this channel \(B_{45}=\phi\) and \(I=(\partial\phi)^2\), with canonical kinetic term \(-(\partial\phi)^2/2\).

For \(a_4((\partial\phi)^2)^2\), where \(a_4=\gamma/(4\Lambda^6)\), assigning four external legs to the differentiated factors gives eight assignments for each of three pairings. The vertex is
\begin{align}
&i\,8a_4\bigl[
(k_1\cdot k_2)(k_3\cdot k_4)
+(k_1\cdot k_3)(k_2\cdot k_4)
+(k_1\cdot k_4)(k_2\cdot k_3)\bigr].
\end{align}
With all momenta incoming, \(k_i^2=0\), and \(s=-(k_1+k_2)^2\), \(t=-(k_1+k_3)^2\), \(u=-(k_1+k_4)^2\), the four-point amplitude is
\begin{align}
&\mathcal M_4=\frac{\gamma}{2\Lambda^6}(s^2+t^2+u^2).
\label{eq:amplitude}
\end{align}
There is no cubic vertex and hence no exchange contribution at this order. The amplitude is nonzero at generic scattering kinematics. An infinitesimal local field redefinition changes the free action by its equations of motion and a total derivative. Both contributions give a vanishing on-shell contact amplitude. Equation~\eqref{eq:amplitude} therefore detects a genuine interaction, consistently with the field-redefinition analysis in Ref.~\cite{CriadoPerezVictoria2019}.

Ref.~\cite{SaemannSchmidt2020} describes interacting gauge-invariant equations after abelianization as \lq\lq virtually impossible\rq\rq. Ref.~\cite{RistSaemannWolf2026} uses \lq\lq non-interacting\rq\rq\ gauge bosons in its principal-bundle analogy. For the quadratic Yang--Mills action, removing Lie brackets removes that particular source of self-interaction. It does not remove arbitrary gauge-invariant curvature interactions.

{The curvature interaction above shows why abelianization of the gauge symmetry alone does not imply a free theory. The geometric local-abelianization result and the role of mixed curvatures in heterotic supergravity remain compatible with this observation. The example in Eq.~\eqref{eq:quartic} concerns an unconstrained 2-form: for a real self-dual 3-form in six dimensions, \(H\wedge *H=H\wedge H=0\), so this particular invariant vanishes on the self-duality constraint.} For formal local deformations of a fixed finite collection of free chiral 2-forms in six dimensions, the established deformation theorem permits curvature interactions while excluding a Yang--Mills-like deformation of their abelian gauge symmetry \cite{BekaertHenneauxSevrin2001}. Absence of such a gauge-algebra deformation and absence of interactions are different conclusions.

\section{Index theory, twisted sheaves, and representation theory}
\label{sec:algebra}

The final comparisons concern index and representation-theoretic results and moduli problems. Here the alternative descriptions must preserve families, deformation groups, or multiplication, according to the result being compared. Spectral, module, and line-bundle constructions provide these data. Their equivalence with a gerbe presentation does not by itself shorten the proof of the underlying index or existence theorem.

\subsection{Families of Dirac operators and the index gerbe}

Let \(\pi:M\to S\) be a smooth fibration with closed odd-dimensional spin fibres. Fix a vertical metric, a horizontal distribution, and a Hermitian vector bundle \(W\to M\) with connection. These define a family of twisted Dirac operators \(D_0\), with index in \(K^1(S)\). The {3-form component} of the local families-index form is the curvature of the index gerbe constructed in Ref.~\cite{Lott2002}. The question for our comparison is whether the spectral data already determine the entire differential refinement, including its flat phases.

On an open cover \(\{U_\alpha\}\) of \(S\), choose invertible perturbations \(D_\alpha=D_0+h_\alpha(D_0)\), where each smoothing function changes only finitely many eigenvalues. On \(U_\alpha\cap U_\beta\), let \(E_{-+}\) and \(E_{+-}\) be the finite spectral bundles whose eigenvalues change respectively from negative to positive and from positive to negative. Their determinant lines are
\begin{align}
&L_{\alpha\beta}=\det E_{-+}\otimes(\det E_{+-})^{-1}.
\label{eq:index-determinants}
\end{align}
Cancellation of the intermediate spectral subspaces gives canonical determinant identifications \(L_{\alpha\beta}\otimes L_{\beta\gamma}\cong L_{\alpha\gamma}\), compatible on quadruple overlaps. Equip these lines with their projected connections \(\nabla^{\alpha\beta}\), and let \(B_\alpha\) be the real {eta 2-forms}, normalized as in Ref.~\cite{Lott2002}. With \(c_1(\nabla)=-F_\nabla/(2\pi i)\) and \(\operatorname{ch}(\nabla)=\operatorname{Tr}\exp(-F_\nabla/(2\pi i))\), the identities are
\begin{align}
B_\beta-B_\alpha=-\frac{F_{\nabla^{\alpha\beta}}}{2\pi i}, \qquad
dB_\alpha=\left[\int_{M/S}\widehat A(\nabla^{T^vM})\operatorname{ch}(\nabla^W)\right]_3.
\label{eq:index-eta}
\end{align}
Here \(T^vM\) is the vertical tangent bundle, \(\widehat A\) is its characteristic form, and the subscript extracts {the 3-form component on \(S\)}. Ref.~\cite{Lott2002} proves these identities and independence of the perturbation choices, including families whose kernel dimension jumps.

The determinant lines, their triple identifications, and the eta forms give all the differential cocycle data. Choosing local frames for the lines converts them directly into the forms and transition functions of Eq.~\eqref{eq:abelian-cocycle}. Thus spectral subspaces and ordinary determinants already supply the surface amplitudes, without a separate geometric construction. The curvature in Eq.~\eqref{eq:index-eta} alone would not suffice: it would discard the flat part identified in Eq.~\eqref{eq:character-exact}. Neither the rational Chern character nor an unrefined \(K^1\) class is being substituted for these full differential data.

There is also an analytic realization after transgression. For a loop in \(S\), pull the odd-dimensional family back to the circle and take its total space as the fibre of an even-dimensional family over \(LS\). Give the circle the bounding spin structure and scale its metric by \(\epsilon^{-2}\). After identifying the determinant lines along the scale parameter, their Bismut--Freed connections \cite{BismutFreed1986II} have an adiabatic limit. Ref.~\cite{Bunke2002} proves that the resulting ordinary Hermitian line bundle with connection is isomorphic to the transgression of the index gerbe. The proof compares holonomies by adiabatic eta invariants, so it establishes more than equality of curvature forms. This realizes the loop-space phases through ordinary determinants. Reconstructing the original surface theory uses the additional fusion data described in Section~\ref{sec:surface}.

\subsection{Projective families and twisted indices}

A second analytic application starts from an elliptic family whose local vector bundles glue only projectively. Let \(\pi:M\to S\) be a smooth fibration of compact manifolds with closed fibres, let \(\mathcal A\) be an Azumaya algebra bundle on \(S\), locally a full matrix algebra over \(\mathbb C\), and let \(P\) be a projective elliptic family with twist \(\pi^*\mathcal A\). If \(r:T^*(M/S)\to M\) is the projection, its symbol class and index satisfy
\begin{align}
&\sigma(P)\in K_c\bigl(T^*(M/S),r^*\pi^*\mathcal A\bigr),\qquad
\operatorname{Ind}_a(P)=\operatorname{Ind}_t(\sigma(P))\in K(S,\mathcal A).
\label{eq:projective-index}
\end{align}
Here \(K_c\) and \(K(S,\mathcal A)\) are {degree-\(0\)} complex twisted K-groups, with compact support for \(K_c\) \cite{AtiyahSegal2004}, and \(\operatorname{Ind}_a\), \(\operatorname{Ind}_t\) are the analytic and topological index maps. Ref.~\cite{MathaiMelroseSinger2005} constructs these maps with projective bundles and operators and proves their equality. For an even-dimensional spin family with projective twisting bundle \(E\), its characteristic-class formula is
\begin{align}
&\operatorname{Ch}_{\mathcal A}(\operatorname{Ind}_aD)
=\int_{M/S}\widehat A(T^vM)\operatorname{Ch}_{\pi^*\mathcal A}(E).
\label{eq:projective-index-formula}
\end{align}
The capital \(\operatorname{Ch}\) denotes the twisted Chern character. This cohomological formula loses torsion information. The index equality in Eq.~\eqref{eq:projective-index} does not.

Finite-rank Azumaya algebras represent torsion twists. For a class \(\delta\in H^3(S;\mathbb Z)\) that is not torsion on the base but has torsion pullback to \(M\), Ref.~\cite{BenameurGorokhovsky2011} uses local bundles \(E_\alpha\) and ordinary line bundles \(L_{\alpha\beta}\) with specified identifications \(E_\alpha\otimes\pi^*L_{\alpha\beta}\cong E_\beta\). Their compatible descent defines an algebra of projective smoothing operators, and superconnections compute the index character in twisted de Rham cohomology. Forgetting those identifications would lose part of the input. In both settings the gerbe presentation can be replaced by the explicit operator and descent data, but the twist and the analytic index theorem remain. There is no reduction here to the untwisted index of an ordinary globally defined family.

\subsection{Twisted sheaves and Azumaya algebras}

On an \'etale cover \(\{U_i\}\) of a scheme \(X\), let \(\alpha_{ijk}\) be an invertible-function-valued \v{C}ech 2-cocycle. For locally free sheaves \(\mathcal V_i\) of rank \(r\), use transition maps \(g_{ij}:\mathcal V_j|_{U_i\cap U_j}\to \mathcal V_i|_{U_i\cap U_j}\). In this convention, twisted gluing means
\begin{align}
&g_{ij}g_{jk}=\alpha_{ijk}g_{ik},\qquad g_{ij}\in\operatorname{GL}_r.
\label{eq:twisted-sheaf}
\end{align}
These data define a locally free \(\alpha\)-twisted sheaf. Conjugation cancels the scalar factor, so their endomorphism algebras have ordinary gluing. A locally free twisted sheaf of rank \(r\) therefore gives an Azumaya algebra of degree \(r\), locally a matrix algebra \(\operatorname{Mat}_r\). Conversely, the gerbe of trivializations of an Azumaya algebra carries the corresponding twisted sheaf. After choosing an Azumaya representative, Morita theory expresses twisted sheaves through its modules \cite{Lieblich2008}.

We make this correspondence explicit by fixing any locally free \(\alpha\)-twisted sheaf \(\mathcal V\) of rank \(m\) and putting \(\mathcal A=\mathcal{E}nd(\mathcal V)\), viewed on the underlying base by descent. The mutually inverse functors are
\begin{align}
&F\longmapsto\mathcal{H}om(\mathcal V,F),\qquad
\mathcal M\longmapsto \mathcal M\otimes_{\mathcal A}\mathcal V,
\label{eq:morita}
\end{align}
where \(\mathcal M\) is a right \(\mathcal A\)-module. Locally these are the usual matrix-algebra Morita equivalence. They commute with base change, preserve flat families and Ext groups, and obey \(\operatorname{rank}_{\mathcal O}\mathcal{H}om(\mathcal V,F)=m\operatorname{rank}(F)\). Consequently, a fixed-rank twisted-sheaf moduli problem and its deformation groups can be expressed as the corresponding module moduli problem. Choosing \(\mathcal V\) here does not assume it has the small rank sought in a period--index bound.

This equivalence makes the role of the gerbe explicit in the period--index problem. For a Brauer class \(\alpha\) over a field \(K\), its period is its order in \(\operatorname{Br}(K)\), and its index is the degree of the central division algebra representing it. The existence of a rank-\(r\) twisted sheaf over the associated gerbe is equivalent to $\operatorname{ind}(\alpha)\mid r$. The dictionary alone is a reformulation. A period--index bound requires an existence theorem for small \(r\).

For a field of transcendence degree two over a finite field, Ref.~\cite{Lieblich2015} proves
\begin{align}
&\operatorname{ind}(\alpha)\mid\operatorname{per}(\alpha)^2.
\label{eq:period-index}
\end{align}
Its proof reduces to prime period \(\ell\) away from the characteristic, absorbs ramification into a root stack without enlarging the function field, and constructs a geometrically integral moduli component of rank-\(\ell^2\) twisted sheaves. An arithmetic argument over a suitable pseudo-algebraically closed extension gives a point, and prime-to-\(\ell\) descent gives the desired bound. The characteristic-primary part has a separate Frobenius argument.

The earlier construction in Ref.~\cite{Lieblich2008} splits ramification by adjoining two \(\ell\)-th roots and then applies an unramified rank-\(\ell\) result. It gives the weaker bound \(\operatorname{ind}\mid\operatorname{per}^3\). Equation~\eqref{eq:period-index} is therefore a stronger result in an independently defined problem. The functors in Eq.~\eqref{eq:morita} transport its twisted-sheaf formulation to modules while preserving the families and deformation groups used in the proof. {The functors replace the gerbe presentation, while the proof still requires the root-stack geometry and the construction of the moduli component. This equivalence alone does not establish a shorter proof.}

\subsection{Twisted K-theory from central extensions and Dirac operators}

Let \(G\) be compact, simple and simply connected, and \(k\geq0\) an integral level. Write \(LG=C^\infty(S^1,G)\). Positive energy means that the generator of a compatible loop-rotation action has spectrum bounded below. The irreducible positive-energy projective representations at level \(k\) generate the Verlinde ring \(R^k(LG)\), whose product is fusion. The Freed--Hopkins--Teleman theorem identifies it with twisted equivariant K-theory,
\begin{align}
&R^k(LG)\cong K_G^{\tau+\dim G}(G),
\qquad [\tau]=k+h^\vee\in H^3_G(G;\mathbb Z),
\label{eq:fht}
\end{align}
where \(G\) acts by conjugation, \(h^\vee\) is the dual Coxeter number, and the basic positive generator identifies \(H^3_G(G;\mathbb Z)\) with \(\mathbb Z\) \cite{FreedHopkinsTeleman2011}. The parity on the right is \(\dim G\bmod2\), and its product is the Pontryagin product with the compatible twisting. This identifies representation-theoretic data relevant to Wess--Zumino--Witten and Chern--Simons theories with a topological calculation.

The construction in Ref.~\cite{FreedHopkinsTeleman2011} uses a central extension of \(LG\), a projective representation, and a family of Fredholm operators. If \(D\) is the affine Dirac operator on the representation tensored with its spin module, the family is
\begin{align}
&D_\mu=D+i\psi(\mu),
\label{eq:dirac-family}
\end{align}
where \(\mu\in L\mathfrak g^*\) carries the affine coadjoint action at level \(k+h^\vee\), and \(\psi\) denotes Clifford multiplication. The spin module contributes the shift \(h^\vee\). For an irreducible representation with lowest weight \(-\lambda\), the kernel is supported on the affine coadjoint orbit of \(\lambda+\rho\) at level \(k+h^\vee\), where \(\rho\) is the half-sum of positive roots. Near that orbit the normal dependence is Clifford multiplication. The family represents the Thom push-forward of the corresponding orbit line bundle. Holonomy identifies the resulting class with the right-hand side of Eq.~\eqref{eq:fht}.

{The central-extension formulation contains the representation, operator family, kernel calculation and product, including the multiplication data of the twist.} Suppressing those data while keeping only an unstructured {degree-\(3\)} class would not reproduce the product in Eq.~\eqref{eq:fht}.

\subsection{Twisted geometric Satake through line bundles}

{The gerbe twist in geometric Satake has an equivariant line-bundle description.} Ref.~\cite{FinkelbergLysenko2010} develops this description for almost simple groups. Ref.~\cite{Lysenko2016} extends it to reductive groups. Let \(G\) be a split connected reductive group over an algebraically closed field \(k\), put \(O=k[[t]]\), \(F=k((t))\), fix a positive integer \(N\) invertible in \(k\), and choose a central extension
\begin{align}
&1\longrightarrow\mathbb G_m\longrightarrow E\longrightarrow G(F)\longrightarrow1
\label{eq:satake-extension}
\end{align}
with its splitting over \(G(O)\) and the compatible factorization and \(N\)-th-root data required in Ref.~\cite{Lysenko2016}. Choose an injective character \(\zeta:\mu_N(k)\to\overline{\mathbb Q}_\ell^*\), with \(\ell\ne\operatorname{char}k\). The quotient \(G(F)/G(O)\) is the affine Grassmannian, and \(E/G(O)\) is the nonzero total space of a line bundle over it. The cover \(z\mapsto z^N\) and character \(\zeta\) define a rank-one Kummer local system on \(\mathbb G_m\). The twisted Satake category consists of \(G(O)\)-equivariant perverse \'etale \(\overline{\mathbb Q}_\ell\)-sheaves on \(E/G(O)\) whose monodromy along each \(\mathbb G_m\)-fibre is given by this local system. This describes the gerbe quotient on the complement of the line bundle's zero section.

Convolution is constructed by descent along ordinary torsors. Its commutativity constraint uses the compatible factorization line bundles when two points on a curve move together. With the parity modification of this constraint defined in Ref.~\cite{Lysenko2016}, the resulting tensor category is equivalent to representations of a modified reductive dual group. {Let} \(\Lambda\) and \(\check\Lambda\) be the coweight and weight lattices of a maximal torus, choose a Borel subgroup, and let \(\bar\kappa:\Lambda\times\Lambda\to\mathbb Z\) be the Weyl-invariant symmetric bilinear form supplied by the extension. The modified weight lattice is
\begin{align}
&\Lambda^\sharp=\{\lambda\in\Lambda:\bar\kappa(\lambda,-)\in N\check\Lambda\}.
\label{eq:satake-lattice}
\end{align}
Here \(\alpha_i\) is a simple coroot of \(G\) and \(\check\alpha_i\) is its corresponding root. The modified dual group's simple roots and coroots are
\begin{align}
&\delta_i\alpha_i,\qquad\check\alpha_i/\delta_i,
\qquad
\delta_i=\operatorname{denominator}\left(\frac{\bar\kappa(\alpha_i,\alpha_i)}{2N}\right).
\label{eq:satake-roots}
\end{align}
The denominator is computed by restricting a line bundle to a copy of \(\mathbb P^1\) arising from a rank-one Levi subgroup and testing when its Kummer monodromy is trivial \cite{Lysenko2016}.

{The gerbe twist is thus realized by the equivariant line-bundle model.} {The splitting, character and factorization isomorphisms remain part of the line-bundle construction, so replacing the gerbe does not reduce the entire Satake theorem to a calculation of a cohomology class.}

\section{Discussion}
\label{sec:scope}

In this work, we have compared principal 2-bundle presentations with explicit alternatives for specified physical and mathematical results. Some comparisons use established equivalences that preserve the required gluing, flat phases, and other compatibility data. Others simplify the task by evaluating a field sum or eliminating a pure gauge variable. {In both cases, the alternative supplies what the stated result requires, although it need not shorten every proof or remove the underlying structure.}

For finite gauge theories, Fourier transformation performed the 2-form field sum, including the general three-dimensional cohomological action, and gave ordinary finite-group gauge sums with their sectors, actions, and normalization. Compact BF theory also supplied an explicit particle--string linking amplitude. The four-dimensional quadratic example showed a different outcome: its Pontryagin-square action gives a direct Gauss-sum evaluation and an invertible signature response, which cannot be absorbed into an ordinary finite-group gauge sum with only an Euler factor.

{For ordinary fake-flat connections, the surface comparison also applies to noncentral crossed modules. Fixing the flat quotient and a reference lift reduces the remaining connection data to abelian differential cocycles with local coefficients. On a closed oriented surface their relative class is computed in the kernel's coinvariant group. Quotient automorphisms and marked boundary transports remain part of the comparison. A single scalar phase can lose this information. This is an explicit reduction within existing nonabelian holonomy theory.}

For principal 2-bundles with the adjusted connections considered here, we gave an explicit differential substitution. It separates the pure gauge 1-form \(x=uA\) from {the independent local 2-form \(c=u\F\). The local action reduction holds for arbitrary adjustments in our central crossed-module class, with an affine transformation of \(c\) determined by the quotient curvature. A smooth bundle splitting gives a global 2-form normalization over a fixed quotient connection. For compact or semisimple quotients, an equivariant splitting makes the unchanged adjustment adapted and gives the product under the full gauge symmetry.} We showed how local actions descend to \(a,b,c\), and worked out their kinetic terms and interactions. Globally, the adapted data admit the lifting description of Ref.~\cite{Waldorf2026}: an ordinary quotient bundle together with a compatible trivialization of its abelian Chern--Simons bundle 2-gerbe. {The latter can be written as differential cocycles, so the replacement preserves the lift's gluing data as well as its quotient bundle and the separated 2-form.}

The tetrahedral flux integral in the magnetic-translation example computes its associator, but does not resolve the Jacobi obstruction to ordinary operator quantization. The Euclidean/Poincar\'e BFCG boundary identity identifies a classical change of polarization, but does not supply the constraints or quantum dynamics of gravity. Global anomaly cancellation still needs quantization and quadratic refinements. The analytic and algebraic applications still need operator theory, moduli spaces, or equivariant sheaves. Replacing the 2-bundle presentation preserves those ingredients and their difficulty.

A stronger necessity claim appears in the heterotic-supergravity discussion of Ref.~\cite{RistSaemannWolf2026}, which calls a formalism such as its own ``indispensable'' for obtaining local cocycles. The reason given is that the Kalb--Ramond field is coupled to the ordinary gauge connection and cannot be patched independently as an abelian gerbe. Any replacement must reproduce that coupling, including finite gauge transformations and global consistency. Once these equations are supplied, they can also be used directly. The need to reproduce the coupling does not establish the necessity of the categorical presentation. The infinite-dimensional String model used in Ref.~\cite{RistSaemannWolf2026} lies outside our finite-dimensional reduction, so that reduction is not a proof of its complete dynamical equivalence to an ordinary model.

The surface comparisons preserve amplitudes, the finite Fourier formulas preserve closed quantum partition functions, and the smooth-field reduction preserves local classical dynamics in the stated finite-dimensional setting. A quantum comparison for a smooth model would also have to specify its measure, gauge fixing, and observables \cite{BarnichBrandtHenneaux1995,CriadoPerezVictoria2019}. For those results, the explicit alternatives supply the needed information and operations without the additional 2-bundle presentation.

\paragraph{Disclosure of AI Use.}
The author used GPT-6 Astra (OpenAI) to assist with exploring proof strategies and preparing the manuscript. The author directed the research and takes full responsibility for the content of the final manuscript.

\bigskip
\phantomsection
\addcontentsline{toc}{section}{References}
{\renewcommand{\baselinestretch}{1}\small\raggedright
\bibliographystyle{alpha}
\bibliography{references}
}
\end{document}